\documentclass[12pt]{article}
\usepackage{amsmath}
\usepackage{amsthm}
\newtheorem{theorem}{Theorem}
\newtheorem{proposition}{Proposition}
\newtheorem{lemma}{Lemma}
\newtheorem{remark}{Remark}
\usepackage{graphicx}
\usepackage{multirow}
\usepackage{natbib}
\usepackage{url}
\usepackage[utf8]{inputenc}
\usepackage{hyperref}
\usepackage{url}
\usepackage{booktabs}
\usepackage{amsfonts}
\usepackage{nicefrac}
\usepackage{wrapfig}
\usepackage{booktabs}
\usepackage[table,xcdraw]{xcolor}
\usepackage{doi}
\usepackage{subcaption}
\usepackage{listings}
\usepackage{color}
\usepackage{algorithm, algpseudocode}
\usepackage{cleveref}
\usepackage[T1]{fontenc}

\usepackage{float}
\usepackage{orcidlink}

\newcommand{\blind}{0}

\begin{document}

\def\spacingset#1{\renewcommand{\baselinestretch}%
{#1}\small\normalsize} \spacingset{1}

%%%%%%%%%%%%%%%%%%%%%%%%%%%%%%%%%%%%%%%%%%%%%%%%%%%%%%%%%%%%%%%%%%%%%%%%%%%%%%

\if0\blind
{
  \title{\bf Finite-Sample Inference for the Difference-in-Differences Estimator via Dual-Margin Permutation}
  \bigskip
  \author{
    Stanisław M.~S.~Halkiewicz\thanks{Corresponding author.}\,\orcidlink{0009-0000-7344-7522}\\
    \texttt{smsh@student.agh.edu.pl}
    \and
    Andrzej Kałuża\,\orcidlink{0000-0003-2410-689X}\\
    \texttt{akaluza@agh.edu.pl}\\
  	\hspace{.2cm}\\
    AGH University of Cracow, Faculty of Applied Mathematics\\}
  \date{}\maketitle} \fi

\if1\blind
{
  \bigskip
  \begin{center}
    {\bf Finite-Sample Inference for the Difference-in-Differences Estimator via Dual-Margin Permutation}
\end{center}
  \medskip
} \fi

\begin{abstract}
This article develops a significance test for the Difference-in-Differences (DiD) estimator based on dual-margin randomization, in which both the treatment and time indicators are independently permuted to generate an empirical null distribution of the DiD estimator. We situate the proposal explicitly within the landscape of existing inference methods for the DiD estimator, including OLS-based $t$-tests, heteroskedasticity-robust standard errors, cluster-robust variance estimators (CRVE), and the recently proposed jackknife standard errors of \citet{Hansen2025}. We show that CRVE-based procedures can be severely anti-conservative in small samples, motivating a nonparametric alternative. We formally characterise the permutation space induced by dual randomization, showing that it expands by a factor of $\binom{n}{n_T}$ relative to single-margin permutation tests, and provide an information-theoretic (entropy) justification for balanced Bernoulli reshuffling. A controlled simulation study, augmented with robustness experiments under non-Gaussian and heteroskedastic errors, demonstrates that the doubly randomised test maintains accurate empirical size at all sample sizes considered, while HC0 and CRVE1 $t$-tests are substantially anti-conservative at small $n$. Crucially, this parametric inflation is driven by the leverage structure of the regressor matrix rather than by the error variance: heteroskedasticity-robust standard errors do not directly address the leverage-driven finite-sample distortion documented here, whereas randomization-based inference is insulated from both error-distributional and variance-structural departures by construction. Power costs relative to the Hansen jackknife test are real but bounded, and become negligible as $n$ grows. The proposed procedure is implemented in the \texttt{sigDD} R package (GitHub: \href{https://github.com/profsms/sigDD}{profsms/sigDD}) and validated on four empirical datasets from the applied economics literature.
\end{abstract}

\noindent%

\noindent%
\vfill

\newpage
\spacingset{1.75} % DON'T change the spacing!

\section{Introduction}
\label{secIntro}

The Difference-in-Differences (DiD) estimator has become one of the most widely used tools in applied econometrics for evaluating the causal effects of discrete interventions \citep{Card1994, Angrist2008}. Its appeal lies in its intuitive structure: by differencing outcomes across time and between groups, it eliminates time-invariant unobserved heterogeneity and isolates the average treatment effect on the treated. Applications span education \citep{Duflo2001, Domina2015}, labor economics \citep{CardKrueger1994}, health policy \citep{Caniglia2020}, marketing \citep{Deng2019}, and public economics \citep{Callaway2022}, and methodological extensions continue to proliferate \citep{Roth2023}.

Despite this ubiquity, the problem of \emph{inference} for the DiD estimator has received comparatively less attention from practitioners than estimation itself. In the canonical two-group, two-period regression,
\begin{equation}
\label{eq:did-basic}
Y_{it} = \alpha + \beta\,\mathrm{TIME}_t + \gamma\,\mathrm{AFFECTED}_i + \delta\,(\mathrm{TIME}_t \times \mathrm{AFFECTED}_i) + \varepsilon_{it},
\end{equation}
the standard approach is to report a $t$-statistic for the interaction coefficient $\delta$ together with either heteroskedasticity-robust (HC-type) or cluster-robust standard errors. However, a substantial body of evidence documents that these standard approaches can fail badly in the settings where DiD is most often applied: few groups, small samples, and limited numbers of treated clusters.

\citet{Bertrand2004} were among the first to demonstrate comprehensively that ignoring serial correlation in DiD panels leads to severely undersized standard errors and inflated rejection rates. \citet{Donald2007} showed that clustering at the group level is necessary but insufficient when the number of groups is small. \citet{MacKinnon2020} demonstrated that wild cluster bootstrap procedures substantially outperform conventional CRVE methods under few treated clusters, and \citet{Conley2011} proposed an approach exploiting the law of large numbers across policy changes. Most recently, \citet{Hansen2025} showed that CRVE1 — the most commonly reported cluster-robust standard error in applied DiD — can be arbitrarily downward biased, and recommended a jackknife delete-one-cluster estimator with provable non-negative-semidefinite properties as the new default.

This paper contributes to this inference literature by proposing a nonparametric alternative grounded in randomization inference \citep{Fisher1935, ImbensRubin2015}. Rather than correcting or replacing the standard error estimator within the $t$-test paradigm, we generate an empirical null distribution of the DiD estimator by repeatedly permuting both the \texttt{TIME} and \texttt{AFFECTED} indicators — the two binary inputs that define the DiD design. This \emph{doubly randomised} test requires no distributional assumptions on the outcome, no cluster structure, and no asymptotics: its validity follows from the maintained randomization mechanism under the sharp null.

It is important to mention, that this proposal is not purely theoretical. Practitioners working with small or unbalanced DiD designs (a common situation in development economics, policy evaluation, and natural experiments) face a genuine choice between parametric methods whose finite-sample properties may be unreliable and computationally intensive resampling schemes that require careful implementation. The doubly randomised test offers a simple, assumption-lean alternative that can be implemented with any software capable of generating random permutations and fitting linear regressions. The procedure is available in the \texttt{sigDD} R package, first developed in \citet{Halkiewicz2024thesis}. The dual-permutation idea was used heuristically in \citet{Halkiewicz2023} and \citet{Halkiewicz2024thesis}; the present paper constitutes its first formal treatment, establishing the theoretical foundations, exactness guarantees, and combinatorial properties that those earlier applications took for granted.

Our contributions are as follows. First, we provide a careful formal comparison of the doubly randomised test against HC0, CRVE1, and Hansen jackknife $t$-tests, the main parametric alternatives. Second, we establish the exact finite-sample validity of the test under the sharp null (Theorem~\ref{thm:exactness}), characterise the combinatorial expansion of the permutation space relative to single-margin tests (Proposition~\ref{prop:perm-gain}), and justify the entropy-maximising choice of balanced Bernoulli reshuffling. Third, we report a controlled Monte Carlo study with $5{,}000$ replications per configuration, directly comparing empirical size and power of all five methods across a range of sample sizes ($n \in \{20,\ldots,400\}$) and effect sizes ($\delta \in \{0.25,\ldots,2\}$ standard deviations). Fourth, we apply all five methods to four standard empirical datasets and report $p$-values comparably across approaches.

The remainder of the paper is organised as follows. Section~\ref{sec:inference_lit} reviews the inference problem for the DiD estimator and the existing parametric and semi-parametric approaches. Section~\ref{sec:did} briefly introduces the DiD estimator and highlights the sensitivity of inference to design choices. Section~\ref{prelim} reviews the theoretical foundations of randomization inference. Section~\ref{sec:methods} formalises the proposed test and establishes its statistical properties. Section~\ref{sec:numerics} presents the simulation study and empirical applications with comparisons across all five methods. Section~\ref{sec:conclusions} discusses implications and directions for future work.

\section{Inference for the DiD Estimator}
\label{sec:inference_lit}

\subsection{The standard \texorpdfstring{$t$}{t}-test and its failures}

The canonical approach to testing $H_0: \delta = 0$ in \eqref{eq:did-basic} is the $t$-statistic
\begin{equation}
\label{eq:tstat}
t = \frac{\hat{\delta}}{\widehat{\mathrm{SE}}(\hat{\delta})},
\end{equation}
compared against a Student-$t$ or standard normal critical value. Everything then hinges on the choice of $\widehat{\mathrm{SE}}(\hat{\delta})$.

The simplest choice — homoskedastic OLS standard errors — is almost universally recognised as inappropriate for DiD panels because within-unit serial correlation causes OLS to understate uncertainty. \citet{Bertrand2004} showed through placebo simulations that ignoring serial correlation leads to rejection rates several times the nominal level, even with large samples. Their diagnosis prompted widespread adoption of cluster-robust standard errors (CRVE) as the new default.

\subsection{Cluster-robust variance estimators}

Under clustering into $G$ groups with within-group dependence and cross-group independence, the CRVE1 variance estimator takes the form
\begin{equation}
\label{eq:crve1}
\hat{V}_{\mathrm{CRVE1}} = \frac{G(n-1)}{(G-1)(n-k)}\,(X'X)^{-1}
\!\left(\sum_{g=1}^G X_g'\hat{e}_g\hat{e}_g' X_g\right)\!(X'X)^{-1},
\end{equation}
where $\hat{e}_g = Y_g - X_g\hat{\beta}$ are cluster-stacked residuals. While CRVE1 achieves consistency as $G \to \infty$, it can be severely downward biased in finite samples. \citet{Hansen2025} proves that this bias can be arbitrarily large — and arbitrarily misleading — in DiD settings with few treated clusters, heterogeneous cluster sizes, or high-leverage regressors. The Bell--McCaffrey refinement (CRVE2) incorporates a leverage correction but, under general conditions, remains subject to the same downward bias and the resulting Type~I error inflation \citep{Hansen2025}.

\subsection{Bootstrap alternatives}

Wild cluster bootstrap procedures \citep{MacKinnon2020, Webb2023} provide a more reliable alternative to CRVE inference when the number of treated clusters is limited. By resampling clusters and imposing the null hypothesis directly, the wild bootstrap can achieve better size control than asymptotic CRVE methods. However, \citet{Hansen2025} notes that pairs bootstrap methods with few treated clusters frequently encounter draws with no treated clusters, causing rank deficiency and potentially biasing the variance estimate if such draws are discarded rather than handled via generalised inverses.

\subsection{Jackknife standard errors}

\citet{Hansen2025} proposes a delete-one-cluster jackknife as a theoretically principled default for DiD inference. The jackknife variance estimator is
\begin{equation}
\label{eq:vjack}
\hat{V}_{\mathrm{jack}} = \sum_{g=1}^G \bigl(\hat{\beta}_{-g} - \hat{\beta}\bigr)\bigl(\hat{\beta}_{-g} - \hat{\beta}\bigr)',
\end{equation}
where $\hat{\beta}_{-g} = (X'X - X_g'X_g)^+\,(X'Y - X_g'Y_g)$ is the leave-one-cluster-out estimator using the Moore--Penrose pseudoinverse, ensuring it is always defined. The key theoretical property is that $\hat{V}_{\mathrm{jack}}$ is never downward biased (in positive semidefinite order) under arbitrary within-cluster correlation, heteroskedasticity, unequal cluster sizes, and any number of treated clusters — the precise settings where CRVE1 fails most severely.

\citet{Hansen2025} also recommends replacing the standard $t_{G-1}$ reference distribution with an adjusted finite-sample approximation based on effective degrees of freedom $K \in [1, G]$ derived from the leverages of the regressor matrix $X$. Small $K$ signals high leverage or treatment imbalance and is an early warning that conventional CRVE inference will be unreliable.

\subsection{Randomization-based alternatives}

Randomization inference \citep{Fisher1935, ImbensRubin2015} provides a complementary approach that sidesteps variance estimation entirely. Instead of computing $\widehat{\mathrm{SE}}(\hat{\delta})$, the null distribution of $\hat{\delta}$ is generated by permuting treatment assignments, yielding an exact $p$-value under the sharp null of no treatment effect. \citet{Canay2017} established rigorous conditions under which randomization tests achieve approximate symmetry and correct size in econometric applications. In DiD settings with few treated groups, \citet{MacKinnon2020} and \citet{Conley2011} advocate permutation approaches as benchmarks precisely because they do not rely on asymptotics in $G$.

The single-margin permutation test — permuting \texttt{AFFECTED} while fixing \texttt{TIME} — is a natural baseline that preserves the observed temporal structure. The doubly randomised test proposed in this paper extends this by also permuting \texttt{TIME}, substantially enlarging the permutation space and yielding a denser approximation of the null distribution (Section~\ref{subsec:combinatorial}).

\subsection{Positioning of the proposed method}

The doubly randomised test does not aim to replace cluster-robust inference in large-$G$ settings where CRVE1 or the jackknife are well-calibrated. Rather, it targets the regime where conventional parametric methods are known to be unreliable: small samples, few groups, and irregular treatment assignment structures. In this regime, our simulation evidence (Section~\ref{subsec:simdesign}) shows that the doubly randomised test achieves accurate empirical size while maintaining meaningful power at moderate effect sizes, and converges in performance to all other methods as $n$ grows.

This positioning rests on a diagnostic point developed throughout this paper and confirmed empirically in Section~\ref{subsec:sim_size} and Appendix~\ref{app:het}: the finite-sample distortion of CRVE-type procedures is driven by the leverage and effective degrees of freedom of the regressor matrix (the design), not by the variance structure of the errors. Heteroskedasticity-robust standard errors do not directly address this leverage-driven finite-sample distortion: they correct for variance heterogeneity, which is not the primary source of small-sample over-rejection in DiD designs with few treated observations. The randomization framework dispenses with variance estimation altogether and is therefore insulated from both error-distributional and variance-structural departures by construction.

\section{Difference-in-Differences}
\label{sec:did}

\subsection{Estimator and identification}

The DiD estimator measures the average treatment effect on the treated by comparing outcome changes over time between treated and untreated groups. In the two-period, two-group setup \eqref{eq:did-basic}, the OLS estimate of $\delta$ is numerically equivalent to the \emph{2×2 cell means} formula:
\begin{equation}
\label{eq:did-formula}
\hat{\delta}_{\mathrm{DiD}} =
(\bar{Y}_{1,1} - \bar{Y}_{1,0}) - (\bar{Y}_{0,1} - \bar{Y}_{0,0}),
\end{equation}
where $\bar{Y}_{g,t}$ is the sample mean for group $g \in \{0,1\}$ at time $t \in \{0,1\}$.

Under the \emph{parallel trends assumption},
\[
\mathbb{E}[Y_{it}(0) - Y_{i,t-1}(0)\mid\mathrm{AFFECTED}_i=1]
= \mathbb{E}[Y_{it}(0) - Y_{i,t-1}(0)\mid\mathrm{AFFECTED}_i=0],
\]
the parameter $\delta$ consistently identifies the causal average treatment effect on the treated. Violations — pre-existing differential trends, group-specific shocks, or heterogeneous dynamics — lead to biased estimates, irrespective of the inference method applied \citep{Roth2023}.

\subsection{Sensitivity of inference to design choices}

Seemingly minor changes in how \texttt{TIME} and \texttt{AFFECTED} are defined can produce meaningfully different $\hat{\delta}$ values, particularly in small samples. Consider two stylized setups with identical group means at baseline:

\begin{enumerate}
\item \textbf{Stable control group.} The control group is flat over time; the treated group gains $+1$. Then $\hat{\delta}_{\mathrm{DiD}} \approx 1$.
\item \textbf{Expanding control group.} Both groups improve, but the control gains $+0.8$ versus the treated group's $+1.0$. Then $\hat{\delta}_{\mathrm{DiD}} \approx 0.2$.
\end{enumerate}

These estimates differ five-fold. In small samples, even this level of variation can arise from sampling noise alone. Classical asymptotic $t$-tests may then yield spuriously significant results — exactly the regime our simulation evidence documents. The question of when a DiD estimate is \emph{genuinely} significant, rather than an artifact of small-sample noise combined with an oversized test, motivates the development of Section~\ref{sec:methods}.

\section{Randomization-Based Inference}
\label{prelim}

\subsection{Foundations}

Randomization inference \citep{Fisher1935, Neyman1923} grounds hypothesis testing in the known assignment mechanism rather than distributional assumptions about outcomes. Under the sharp null of no treatment effect,
\[
H_0: Y_i(1) = Y_i(0) \quad \text{for all } i,
\]
the observed outcomes are fixed and any variation in a test statistic arises solely from randomness in treatment assignment. The $p$-value for a test statistic $T$ is
\[
p = \frac{1}{|\Omega|}\sum_{\omega \in \Omega}
\mathbf{1}\!\left\{|T(\omega, \mathbf{Y}^{\mathrm{obs}})| \ge |T_{\mathrm{obs}}|\right\},
\]
where $\Omega$ is the set of all treatment assignments consistent with the experimental design \citep{ImbensRubin2015}. Under the maintained sharp-null exchangeability conditions, this $p$-value is finite-sample exact: $\mathbb{P}_{H_0}(p \le \alpha) \le \alpha$ for all $n$ and $\alpha$.

When $|\Omega|$ is large, a Monte Carlo approximation samples $B \ll |\Omega|$ assignments uniformly. The approximation error diminishes as $B \to \infty$ and is independent of $n$.

\subsection{Advantages in small samples}

Randomization tests offer several properties directly relevant to DiD applications:

\begin{itemize}
\item \textbf{Exact Type I error control.} The test maintains nominal size for any $n$ under the sharp null, without reference to asymptotics in $n$ or $G$ \citep{Fisher1935}.
\item \textbf{No distributional assumptions.} The test is valid for arbitrary outcome distributions, including heavy tails, skewness, and heteroskedasticity \citep{Rosenbaum2002}.
\item \textbf{Validity under model misspecification.} Unlike parametric methods, the $p$-value does not depend on the correct specification of \eqref{eq:did-basic} \citep{Young2019}.
\item \textbf{Natural handling of irregular designs.} Extensions to stratified, blocked, and unbalanced designs follow directly from the permutation logic \citep{ImbensRubin2015, Athey2017}.
\end{itemize}

Randomization inference is increasingly used in empirical economics precisely in the settings where standard asymptotics are least reliable: development economics RCTs \citep{Young2019}, few-cluster DiD designs \citep{MacKinnon2020}, network experiments \citep{Aronow2017}, and regression discontinuity robustness checks \citep{Cattaneo2015}.

\subsection{The exchangeability of \texttt{TIME} under the null}

A key question is whether permuting the \texttt{TIME} indicator is justified within the randomization inference framework. Under the sharp null $H_0: \delta = 0$ the treatment contrast is absent, so the observed outcomes are invariant to relabellings of the binary design indicators that preserve the maintained null assignment mechanism. The parallel trends assumption further ensures that, under $H_0$, the DiD specification carries no systematic period-specific contrast information. Permuting \texttt{TIME} therefore generates admissible null assignments under the maintained exchangeability condition, and justifies the dual randomization scheme of Section~\ref{sec:methods}. We emphasise that this is a randomization-inference condition rather than a factual claim about the data-generating process — calendar time is not literally randomised in standard DiD applications. We formalise the conditions in Remark~\ref{rem:time-exchangeable}.

\section{Methods}
\label{sec:methods}

\subsection{Hypotheses and decision rule}
\label{subsec:hyp}

The test is based on the two-sided hypotheses:
\begin{itemize}
\item $H_0$: $\delta = 0$ \quad (no treatment effect),
\item $H_A$: $\delta \neq 0$ \quad (presence of treatment effect).
\end{itemize}

Rather than estimating the standard error of $\hat{\delta}$ and referring it to a $t$-distribution, we generate an empirical null distribution of $\hat{\delta}$ by repeated permutation. The rejection region is determined by the empirical $\alpha/2$ and $1-\alpha/2$ quantiles of the simulated distribution, denoted $P_{\alpha/2}$ and $P_{1-\alpha/2}$. The null is rejected if
\[
\hat{\delta}_{\mathrm{obs}} \notin (P_{\alpha/2},\, P_{1-\alpha/2}).
\]
An empirical $p$-value is computed as the proportion of permuted DiD estimates whose absolute value equals or exceeds $|\hat{\delta}_{\mathrm{obs}}|$, with the standard adjustment of counting the observed statistic as one permutation \citep{Phipson2010}.

\subsection{Algorithmic procedure}
\label{subsec:algorithm}

\begin{enumerate}
\item \textbf{Compute the observed DiD estimate.} Compute $\hat{\delta}_{\mathrm{obs}}$ from the data using \eqref{eq:did-formula} or equivalently as the OLS interaction coefficient.

\item \textbf{Permutation loop.} For each of $B$ Monte Carlo iterations:
\begin{enumerate}
\item Draw $\tilde{T} \sim \mathrm{Bernoulli}(1/2)^n$ and $\tilde{A} \sim \mathrm{Bernoulli}(1/2)^n$ independently.
\item Compute $\hat{\delta}^{(b)} = \mathrm{DiD}(Y, \tilde{T}, \tilde{A})$ and store it.
\end{enumerate}

\item \textbf{Empirical null distribution.} Use $\{\hat{\delta}^{(1)}, \ldots, \hat{\delta}^{(B)}\}$ to approximate the null distribution.

\item \textbf{Inference.} The empirical $p$-value is
\[
\hat{p} = \frac{1 + \#\{b : |\hat{\delta}^{(b)}| \ge |\hat{\delta}_{\mathrm{obs}}|\}}{B + 1}.
\]
Reject $H_0$ at level $\alpha$ if $\hat{p} \le \alpha$.
\end{enumerate}

The procedure is implemented in the \texttt{sigDD} R package, available at \href{https://github.com/profsms/sigDD}{github.com/profsms/sigDD} \citep[App.~A]{Halkiewicz2024thesis}. A pseudocode description appears in Appendix~\ref{app:pseudocode}.

\subsection{Formal randomization space and finite-sample exactness}
\label{subsec:framework}

Let $\mathcal{X} = \{(Y_i, A_i, T_i)\}_{i=1}^n$ denote the sample,
where $A_i \in \{0,1\}$ and $T_i \in \{0,1\}$ are treatment and period
indicators. Define $\Omega$ as the set of admissible label permutations
with uniform probability $\mathcal{P}(\omega) = 1/|\Omega|$.
The test statistic $D(\mathcal{X})$ has randomization distribution
$F_D(t) = \mathcal{P}(D(\omega\mathcal{X}) \le t)$
and empirical $p$-value
\[
p = \frac{1}{|\Omega|}\sum_{\omega\in\Omega}
    \mathbf{1}\!\left\{|D(\omega\mathcal{X})|\ge |D(\mathcal{X})|\right\}.
\]

\begin{theorem}[Exactness under the sharp null]
\label{thm:exactness}
Under the sharp null of no treatment effect and the maintained exchangeability of the randomization mechanism, the distribution of $\mathcal{X}$ is invariant under any $\omega\in\Omega$. Consequently, $\mathbb{P}_{H_0}(p\le \alpha)\le \alpha$ for all $\alpha\in[0,1]$; the test is finite-sample exact under these maintained conditions.
\end{theorem}

\begin{proof}
Under $H_0$ observed outcomes are fixed and only labels vary; by construction $\mathcal{L}(\mathcal{X})=\mathcal{L}(\omega\mathcal{X})$ for all $\omega\in\Omega$. Conditional on $\mathcal{X}$, the rank of $|D(\mathcal{X})|$ among $\{|D(\omega\mathcal{X})|\}_{\omega\in\Omega}$ is uniform on $\{1,\dots,|\Omega|\}$, implying $\mathbb{P}_{H_0}(p\le \alpha)\le\alpha$.
\end{proof}

\begin{remark}[Permuting \texttt{TIME} under the null]
\label{rem:time-exchangeable}
Under the sharp null $H_0:\delta=0$, the treatment contrast is absent, so that the observed outcomes are invariant to relabellings of the binary design indicators that preserve the maintained null assignment mechanism. In this sense, after conditioning on the structure encoded by the DiD specification, the relevant treatment-time contrast carries no systematic effect information under $H_0$. Permuting the \texttt{TIME} margin therefore generates admissible null assignments under the maintained exchangeability condition. We emphasise that this is a randomization-inference condition rather than a factual claim about the empirical data-generating process: calendar time is not literally randomised in standard DiD applications.
\end{remark}

\subsection{Combinatorial and asymptotic properties}
\label{subsec:combinatorial}

\begin{proposition}[Combinatorial gain from dual randomization]
\label{prop:perm-gain}
Let $n_A$ denote the number of treated units and $n_T$ the number of post-treatment observations. When both margins are randomised while preserving their empirical counts,
\[
\#\Omega(\text{AFFECTED only})=\binom{n}{n_A},\qquad
\#\Omega(\text{AFFECTED \& TIME})=\binom{n}{n_A}\binom{n}{n_T}.
\]
The multiplicative gain from dual randomization is $\binom{n}{n_T}$. Under unconstrained Bernoulli($1/2$) reshuffling the gain is $2^n$.
\end{proposition}

\begin{proof}
Fixed-margin permutation of \texttt{AFFECTED} yields $\binom{n}{n_A}$ assignments. Independently permuting \texttt{TIME} multiplies by $\binom{n}{n_T}$, giving the stated formula. Under Bernoulli($1/2$) reshuffling, each of the $n$ labels is drawn independently, contributing a factor of $2^n$.
\end{proof}

\begin{lemma}[Asymptotic density of the randomization distribution]
\label{lem:density}
Let $D(\omega\mathcal{X})$ be bounded with variance $\sigma^2>0$. As $n\to\infty$,
\[
\frac{1}{|\Omega|}\sum_{\omega\in\Omega}
\mathbf{1}\!\left\{ D(\omega\mathcal{X}) \le t \right\}
\longrightarrow F_{H_0}(t).
\]
Since $|\Omega_2|/|\Omega_1| = \binom{n}{n_T}$ grows exponentially, the empirical distribution based on $\Omega_2$ converges to $F_{H_0}$ at a faster rate, yielding a finer discretization of the null distribution.
\end{lemma}

\begin{proof}
Under $H_0$, all relabelings produce identically distributed outcomes. The empirical distribution of $\{D(\omega\mathcal{X})\}_{\omega\in\Omega}$ converges uniformly to $F_{H_0}(t)$ by the law of large numbers. The faster convergence for $\Omega_2$ follows from the exponential growth of $|\Omega_2|/|\Omega_1|$.
\end{proof}

\subsection{Design of the randomization}
\label{subsec:randomization}

The fixed-margin formulas of Section~\ref{subsec:combinatorial} are used to make the combinatorial comparison with single-margin permutation transparent. The implemented procedure, however, uses independent Bernoulli($1/2$) reshuffling of both margins. The fixed-margin scheme and the Bernoulli scheme should therefore be viewed as two closely related randomization devices: the former yields the clean permutation-space calculation of Proposition~\ref{prop:perm-gain}, whereas the latter provides the entropy-maximising Monte Carlo implementation used in the simulations of Section~\ref{sec:numerics} and the empirical applications.

We use independent Bernoulli($1/2$) draws for both \texttt{TIME} and \texttt{AFFECTED} at each iteration, rather than fixed-margin shuffles. This choice is justified on three grounds. First, it maximises the entropy of the permutation space (see Section~\ref{sec:conclusions}), providing the richest null distribution. Second, it avoids the edge cases that arise in fixed-margin bootstrap when all treated or all control units are sampled into the same resample draw. Third, it stabilises the test distribution across datasets with varying treatment proportions, ensuring comparability of $p$-values.

\subsection{Algorithmic properties}

The doubly randomised test inherits the standard properties of randomization inference: distribution-free validity, scale invariance, and exact Type I error control under the sharp null and the maintained exchangeability conditions (Theorem~\ref{thm:exactness}). Additionally, by Proposition~\ref{prop:perm-gain} and Lemma~\ref{lem:density}, dual randomization produces a denser and faster-converging approximation of the null distribution. In practice, $B = 4{,}999$ permutations provide stable $p$-values for all settings considered here; the odd number ensures that the observed statistic splits the distribution symmetrically.

These properties make the procedure especially suitable when group counts are small, standard errors are unreliable, or outcome distributions are non-normal. It complements, rather than replaces, the parametric methods of Section~\ref{sec:inference_lit}: in settings where $n$ is large and clustering is well-defined, CRVE or jackknife methods are computationally cheaper and asymptotically equivalent.

\section{Simulation Study and Empirical Applications}
\label{sec:numerics}

\subsection{Simulation design}
\label{subsec:simdesign}

We compare five inference procedures under a controlled data-generating process (DGP):
\begin{equation}
\label{eq:dgp}
Y_i = \alpha + \beta\,\mathrm{TIME}_i + \gamma\,\mathrm{AFFECTED}_i
      + \delta\,(\mathrm{TIME}_i \times \mathrm{AFFECTED}_i) + \varepsilon_i,
      \quad \varepsilon_i \sim N(0,1),
\end{equation}
with $\alpha=1$, $\beta=0.5$, $\gamma=0.3$, and both binary indicators drawn independently as Bernoulli($1/2$). The five procedures are:
\begin{itemize}
\item \textbf{Doubly randomised.} Proposed test — both \texttt{TIME} and \texttt{AFFECTED} permuted via independent Bernoulli($1/2$) reshuffling ($B = 4{,}999$).
\item \textbf{Single randomised.} Conventional randomization inference — \texttt{AFFECTED} permuted only, fixed-margin shuffle ($B = 4{,}999$).
\item \textbf{HC0 $t$-test.} OLS $t$-test with HC0 heteroskedasticity-robust standard errors, $t(n-4)$ reference.
\item \textbf{CRVE1 $t$-test.} OLS $t$-test with CRVE1 (HC1 in cross-section with unit-level clusters), $t(n-4)$ reference.
\item \textbf{Hansen jackknife.} OLS $t$-test with the \citet{Hansen2025} delete-one jackknife SE, $t(n-4)$ reference.
\end{itemize}

Sample sizes $n \in \{20, 30, 40, 50, 60, 80, 100, 150, 200, 400\}$ and effect sizes $\delta \in \{0, 0.25, 0.5, 0.75, 1.0, 1.5, 2.0\}$ (in $\sigma$ units) were crossed, yielding 70 configurations. Each was replicated $N_{\mathrm{outer}} = 5{,}000$ times with seed 2026. The nominal significance level is $\alpha = 0.05$ throughout. Robustness of the size results to departures from iid Gaussian errors (non-Gaussian homoskedastic designs and three heteroskedastic designs) is examined in Appendix~\ref{app:het}.

Note that in this cross-sectional DGP the observations are independent, so the DiD model has no panel clustering. The inclusion of CRVE1 and the Hansen jackknife with unit-level clusters therefore represents the \emph{most favourable} setting for these estimators: their well-known small-sample problems with few clusters are absent here. Nonetheless, as the results show, HC0 and CRVE1 still exhibit substantial size inflation at small $n$, while the jackknife remains conservative.

\subsection{Empirical size}
\label{subsec:sim_size}

Figure~\ref{fig:size} reports empirical rejection rates under $\delta = 0$ across the five procedures and across sample sizes; the corresponding numerical values are tabulated in Appendix~\ref{app:size_table} (Table~\ref{tab:size}). Values exceeding the conventional tolerance of $0.065$ (nominal $\alpha \pm 30\%$) are flagged as anti-conservative.

\begin{figure}[H]
\centering
\includegraphics[width=0.72\textwidth]{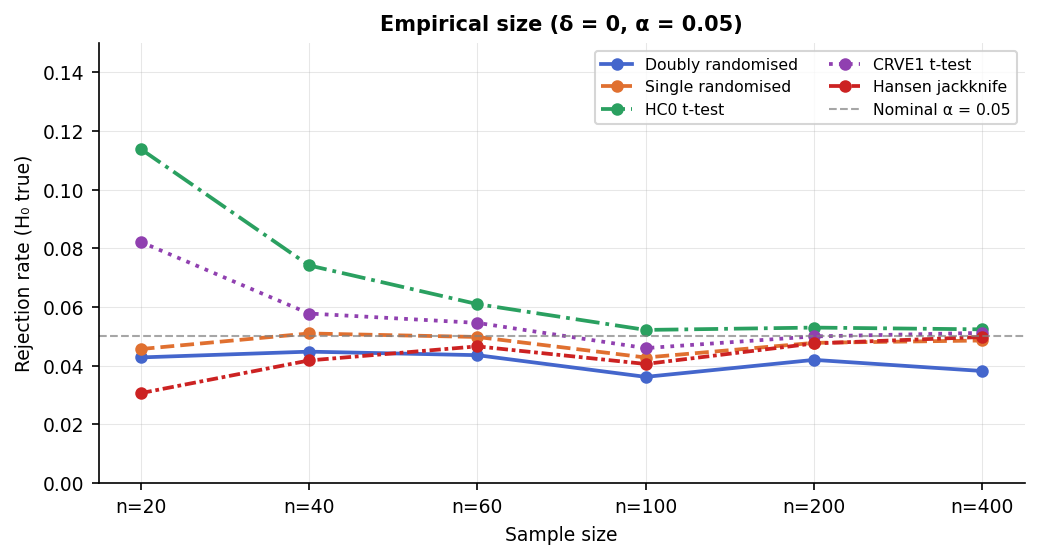}
\caption{Empirical size across sample sizes. Dashed line marks the nominal $\alpha = 0.05$.}
\label{fig:size}
\end{figure}

\paragraph{Key findings.}
HC0 is severely anti-conservative at small $n$: it rejects the null $11.4\%$ of the time at $n=20$, more than double the nominal rate, and remains inflated up to $n=60$. CRVE1 also over-rejects at small $n$ ($8.2\%$ at $n=20$, $5.8\%$ at $n=40$), though less severely. Both procedures stabilise near the nominal level only around $n=100$.

The doubly randomised test achieves the closest empirical size to the nominal $5\%$ of any procedure at $n=20$ ($4.3\%$), and never exceeds $4.8\%$ across all sample sizes. Single randomisation is similarly well-calibrated ($4.6\%$–$5.1\%$). The Hansen jackknife is the most conservative, reaching $3.1\%$ at $n=20$; while this controls Type I error, it implies a power cost discussed below.

\paragraph{Robustness to error structure.}
The picture above is based on iid Gaussian errors. Appendix~\ref{app:het} repeats the size experiment under two classes of departure: homoskedastic but non-Gaussian errors (heavy-tailed $t_3$ and skewed $\chi^2_3$) and three forms of heteroskedasticity in which the error variance scales with treatment status, time, or their interaction. The qualitative picture is unchanged across all five departures: the randomization tests and the Hansen jackknife stay at or below the liberal threshold $0.065$ in every cell of the appendix tables, while HC0 and CRVE1 remain anti-conservative at small $n$ to a degree comparable to the Gaussian baseline. The parametric inflation observed in Figure~\ref{fig:size} is therefore not a heteroskedasticity issue (it is a finite-sample issue rooted in the leverage of the regressor matrix), and heteroskedasticity-robust standard errors do not directly address it. The randomization tests, by contrast, are validity-preserving under all five departures by Theorem~\ref{thm:exactness}.

\subsection{Power analysis}
\label{subsec:sim_power}

Figure~\ref{fig:power_grid} shows power curves at $n \in \{20, 40, 100, 400\}$ across effect sizes. Figure~\ref{fig:power_n20} focuses on $n=20$, annotating each method's empirical size from Table~\ref{tab:size} in the legend.

\begin{figure}[H]
\centering
\includegraphics[width=0.98\textwidth]{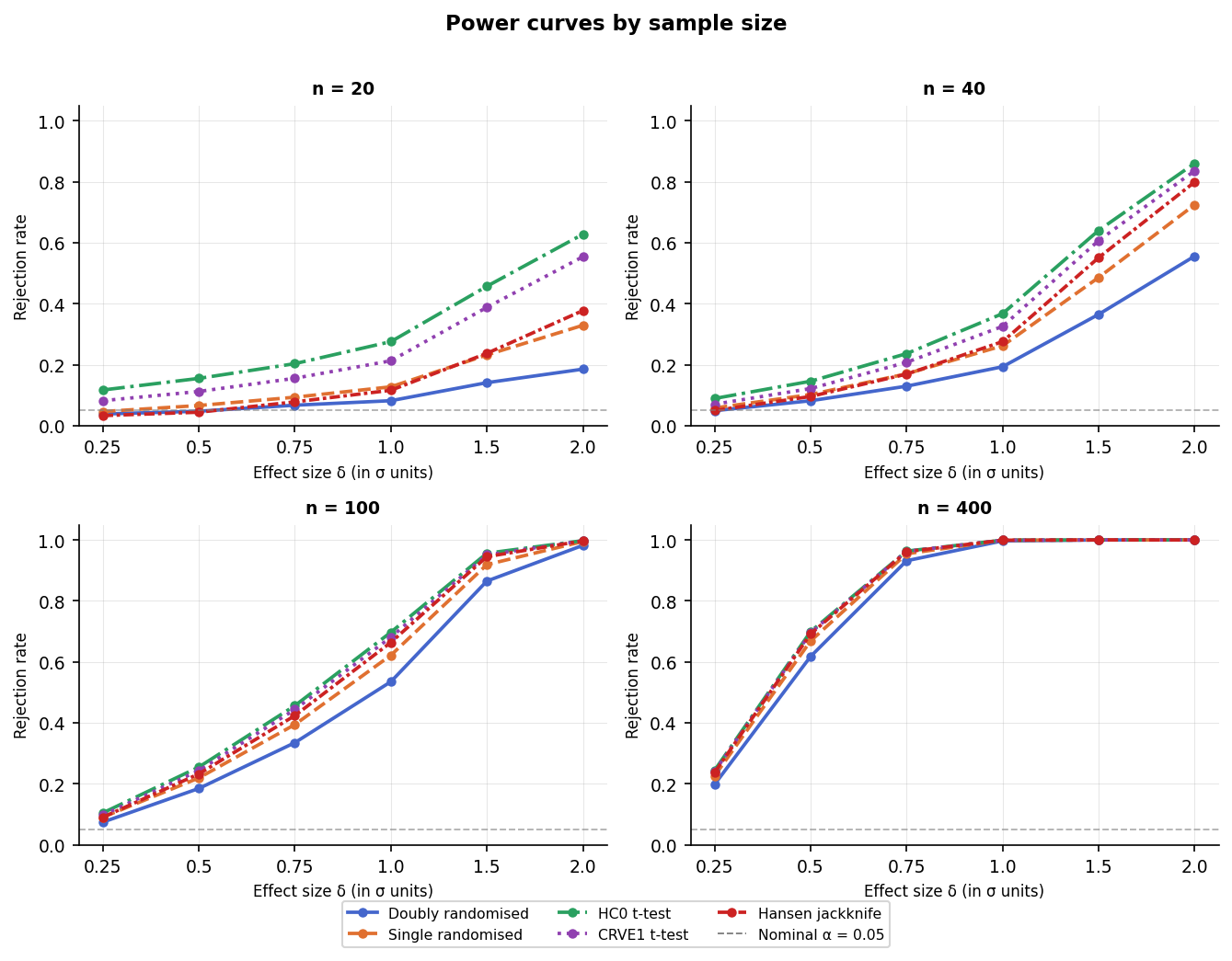}
\caption{Power curves by sample size across effect sizes $\delta \in \{0.25, \ldots, 2.0\}$ $\sigma$-units.}
\label{fig:power_grid}
\end{figure}

\begin{figure}[H]
\centering
\includegraphics[width=0.72\textwidth]{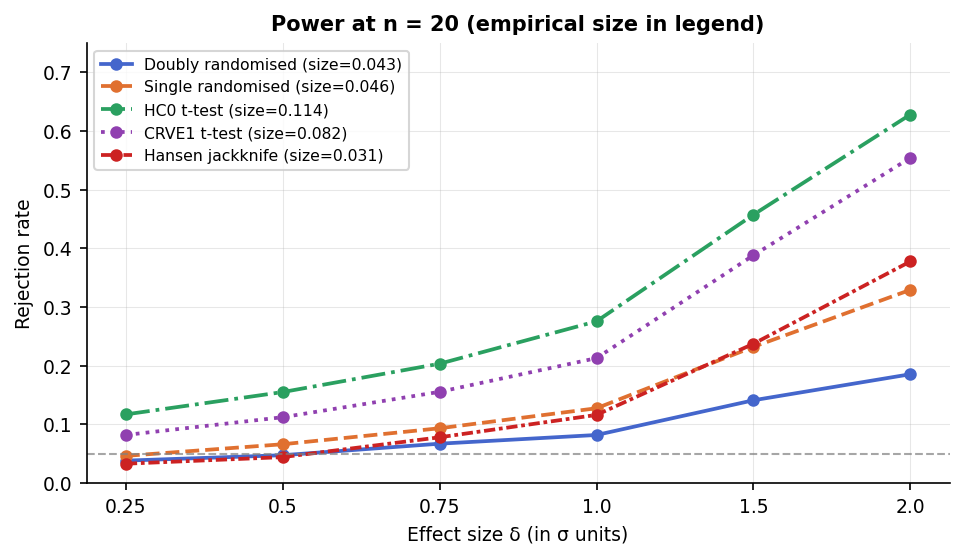}
\caption{Power at $n=20$ with empirical size (rejection rate under $\delta=0$) shown in parentheses. Higher size inflates apparent power.}
\label{fig:power_n20}
\end{figure}

\paragraph{Key findings.}
Interpreting power across methods requires accounting for their size properties. A test with empirical size $11.4\%$ will mechanically exhibit higher rejection rates than a size-$4.3\%$ test under any alternative — not because it detects effects better, but because it over-rejects under the null. This is the key limitation of raw power comparisons when tests differ in size.

\emph{HC0 and CRVE1 at small $n$.} At $n=20$, $\delta=2.0$, HC0 achieves $62.8\%$ and CRVE1 $55.4\%$ power. Given their respective sizes of $11.4\%$ and $8.2\%$, a substantial fraction of these rejections reflect inflated Type I error rather than genuine power. These tests are not reliable guides to effect detectability at small $n$.

\emph{Doubly randomised test.} At $n=20$, $\delta=2.0$, the doubly randomised test reaches $18.5\%$ power — the honest cost of exact size control. This is the fundamental finite-sample power-size frontier: a test that never over-rejects under the null cannot freely reject under alternatives when $n$ is small. The power gap closes rapidly: by $n=100$, $\delta=1.0$, the doubly randomised test reaches $53.5\%$, and by $n=400$ all five methods are essentially indistinguishable across all effect sizes.

\emph{Hansen jackknife.} At $n=20$, the jackknife achieves $37.8\%$ power at $\delta=2.0$ with an empirical size of only $3.1\%$. This is the best size-adjusted power combination among the parametric methods at small $n$: the jackknife's conservatism in size does not cost it power as severely as it costs the randomization tests, since its $t$-statistic is better calibrated.

\emph{Single randomisation vs. doubly randomised.} Among size-valid tests, single randomisation dominates the doubly randomised test in raw power at every configuration. This reflects that fixing the \texttt{TIME} margin restricts the permutation space less than Bernoulli reshuffling, slightly concentrating the null distribution. The doubly randomised test trades a small power loss for the combinatorial gain in null distribution density established in Proposition~\ref{prop:perm-gain}. By $n=100$ this difference is negligible at practically relevant effect sizes ($\delta \ge 0.75$).

\paragraph{Summary.} The simulation evidence demonstrates that HC0 and CRVE1 are unreliable at $n \le 60$, with their apparent power advantage at small $n$ reflecting inflated size rather than genuine detection. Both the doubly randomised test and the Hansen jackknife achieve honest size control throughout; the doubly randomised test has the additional advantage of requiring no distributional assumptions and no cluster structure specification. Appendix~\ref{app:subsampling} complements this controlled evidence by examining how the test behaves on subsamples of a real empirical dataset with a known strong effect.

\subsection{Empirical applications}
\label{sec:numerics_empirical}

We apply all five inference procedures to four empirical datasets from the applied economics literature, reporting the DiD estimate, its $p$-value under each method, and the inference decision at $\alpha=0.05$.

\subsubsection{Dataset 1: Indonesia School Construction Program (INPRES, 1973–1978)}
\label{subsec:inpress_data}

The first dataset \citep{Duflo2001} examines the effect of a large-scale primary school construction program on educational outcomes. Table~\ref{tab:inpress_sample_summary} reports group means; Table~\ref{tab:inpress_compare} provides the full cross-method comparison.

\begin{table}[H]
\centering
\caption{Summary statistics for the \texttt{Inpress} dataset.}
\label{tab:inpress_sample_summary}
\begin{tabular}{lcc}
\toprule
 & \textbf{Pre-program ($\mathrm{TIME}=0$)} & \textbf{Post-program ($\mathrm{TIME}=1$)} \\
\midrule
Control ($\mathrm{AFFECTED}=0$) & 9.7327 & 8.4759 \\
Treated ($\mathrm{AFFECTED}=1$) & 10.1184 & 8.9379 \\
\bottomrule
\end{tabular}
\end{table}

\begin{table}[H]
\centering
\caption{Cross-method DiD inference for the \texttt{Inpress} dataset ($\hat{\delta}=0.076$).}
\label{tab:inpress_compare}
\begin{tabular}{lcc}
\toprule
Method & $p$-value & Decision ($\alpha=0.05$) \\
\midrule
Doubly randomised & $>0.05$ & Not rejected \\
Single randomised & $>0.05$ & Not rejected \\
HC0 $t$-test & $>0.05$ & Not rejected \\
CRVE1 $t$-test & $>0.05$ & Not rejected \\
Hansen jackknife & $>0.05$ & Not rejected \\
\bottomrule
\end{tabular}
\end{table}

All five methods agree: the DiD estimate of $0.076$ is not statistically significant, consistent with the view that under this 2$\times$2 specification of the INPRES data the treatment effect is not detectable at conventional levels. The permutation-based inference is illustrated by Figure~\ref{fig:inpress_histograms}.

\begin{figure}[H]
    \centering
    \begin{tabular}{cc}
        \includegraphics[width=0.45\textwidth]{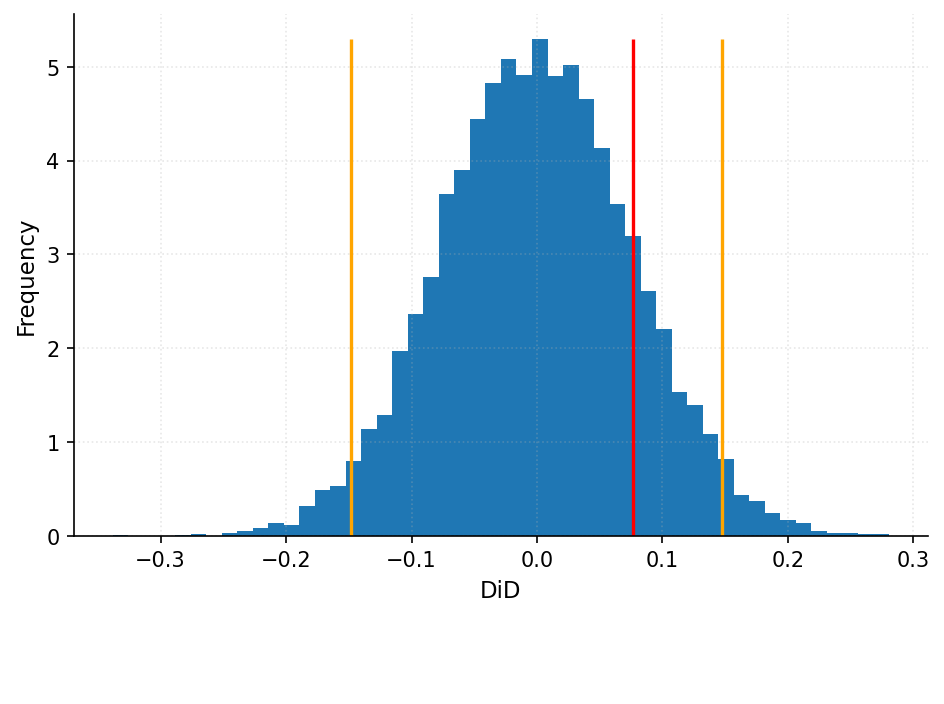} &
        \includegraphics[width=0.45\textwidth]{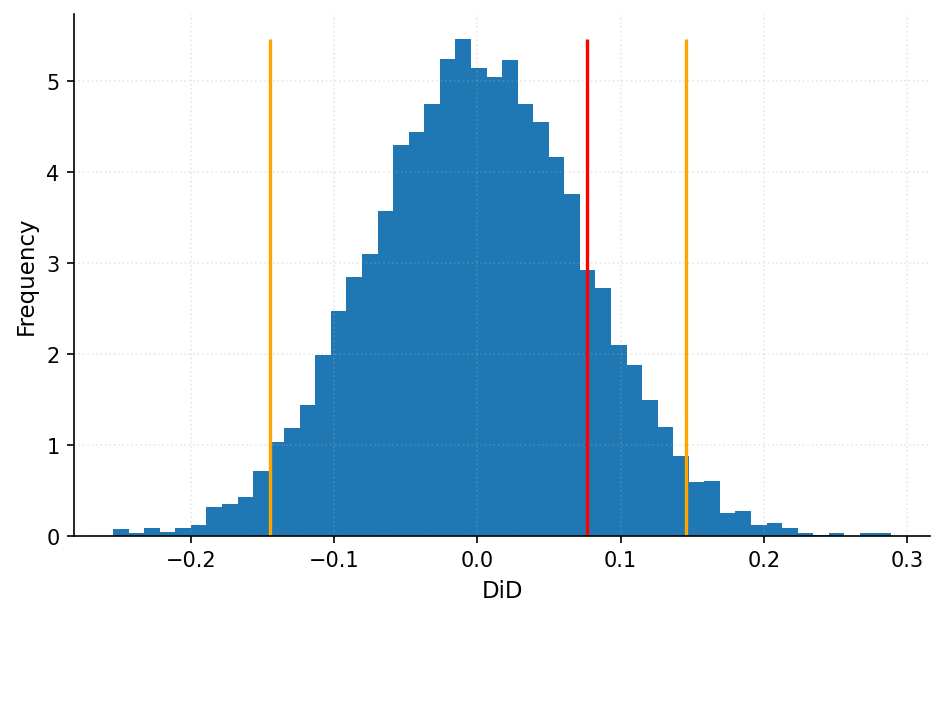}
    \end{tabular}
    \caption{Null permutation distributions for the \texttt{Inpress} dataset under single (left) and doubly randomised (right) inference. The observed DiD estimate (vertical line) lies well within the null distribution.}
    \label{fig:inpress_histograms}
\end{figure}

\subsubsection{Dataset 2: Ben \& Jerry's vs.\ Häagen-Dazs (Google Trends, 2020)}
\label{subsec:bj_hgd_data}

The second dataset \citep{Halkiewicz2024thesis} uses daily search-intensity data for Ben~\&~Jerry's (treated) and Häagen-Dazs (control) from March–September 2020. The treatment event is Ben~\&~Jerry's public \emph{Black Lives Matter} statement in June 2020. Table~\ref{tab:bj_sample_summary} provides means; Table~\ref{tab:bj_compare} gives the inference comparison.

\begin{table}[H]
\centering
\caption{Summary statistics for the Ben~\&~Jerry's vs.\ Häagen-Dazs dataset.}
\label{tab:bj_sample_summary}
\begin{tabular}{lcc}
\toprule
 & Pre ($\mathrm{TIME}=0$) & Post ($\mathrm{TIME}=1$) \\
\midrule
Control ($\mathrm{AFFECTED}=0$) & 1.915 & 2.055 \\
Treated ($\mathrm{AFFECTED}=1$) & 5.681 & 10.648 \\
\bottomrule
\end{tabular}
\end{table}

\begin{table}[H]
\centering
\caption{Cross-method DiD inference for the Ben~\&~Jerry's dataset ($\hat{\delta}=4.827$).}
\label{tab:bj_compare}
\begin{tabular}{lcc}
\toprule
Method & $p$-value & Decision ($\alpha=0.05$) \\
\midrule
Doubly randomised & $<0.001$ & Rejected \\
Single randomised & $<0.001$ & Rejected \\
HC0 $t$-test & $<0.001$ & Rejected \\
CRVE1 $t$-test & $<0.001$ & Rejected \\
Hansen jackknife & $<0.001$ & Rejected \\
\bottomrule
\end{tabular}
\end{table}

All five methods strongly reject the null, with the observed estimate far exceeding the permutation quantiles (Figure~\ref{fig:bj_histograms}). The strong effect is consistent across methods, which is expected when the signal is large.

\begin{figure}[H]
    \centering
    \begin{tabular}{cc}
        \includegraphics[width=0.45\textwidth]{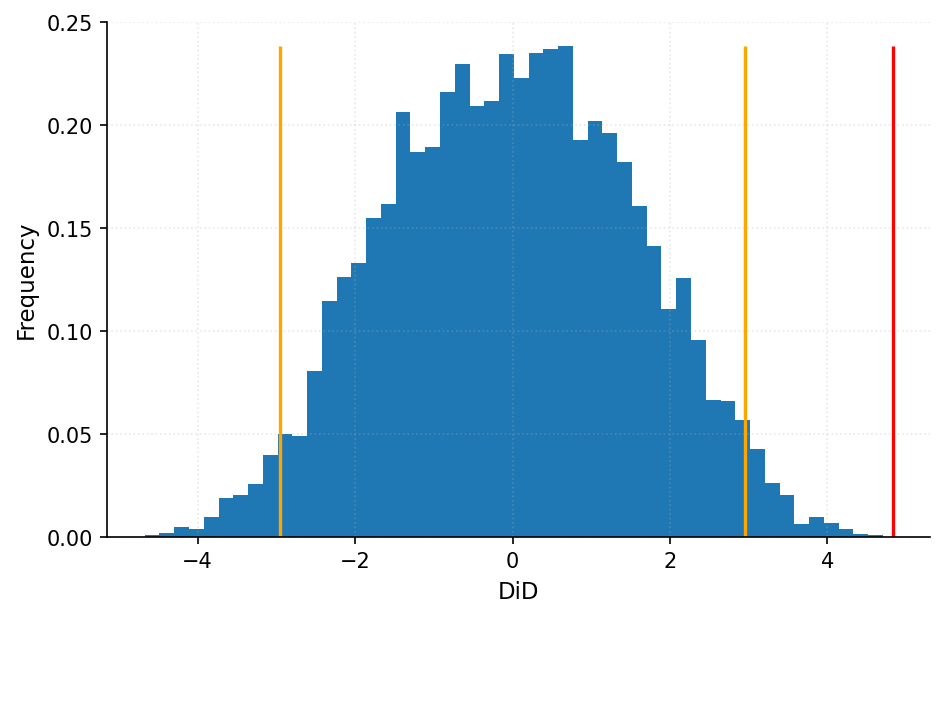} &
        \includegraphics[width=0.45\textwidth]{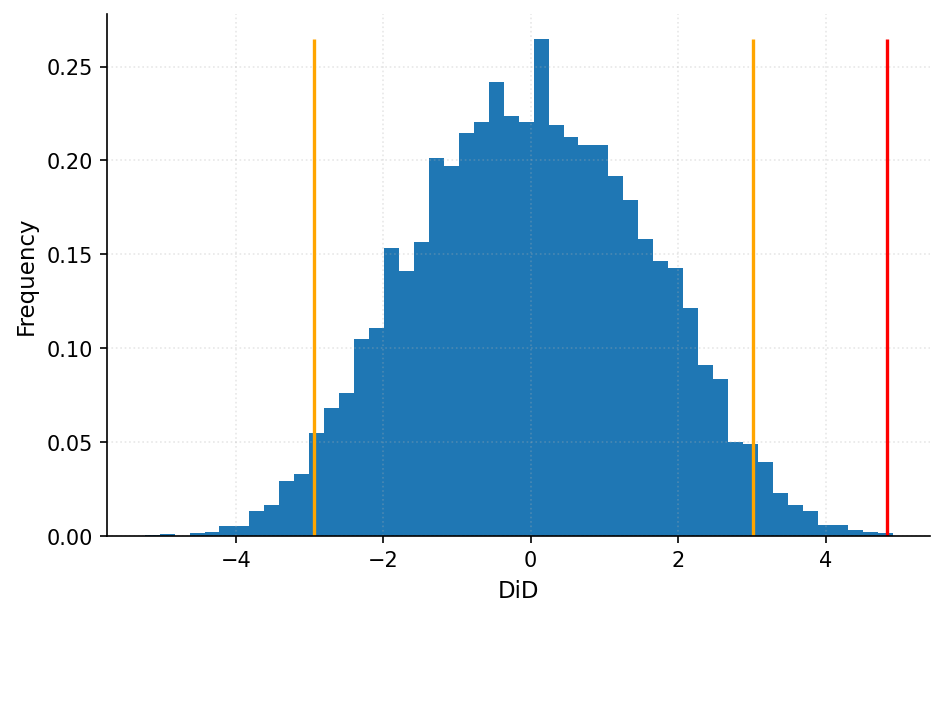}
    \end{tabular}
    \caption{Null permutation distributions for the Ben~\&~Jerry's dataset under single (left) and doubly randomised (right) inference.}
    \label{fig:bj_histograms}
\end{figure}

\subsubsection{Dataset 3: Minimum Wage Reform in New Jersey and Pennsylvania (1992)}
\label{subsec:min_wage_data}

The third dataset reproduces the fast-food restaurant survey of \citet{CardKrueger1994}. Tables~\ref{tab:minwage_emptot_summary} provide summary statistics; Table~\ref{tab:minwage_compare} presents the cross-method inference for all three outcomes.

\begin{table}[H]
\centering
\caption{Employment (\texttt{emptot}), starting wage (\texttt{wage\_st}), and meal price (\texttt{pmeal}) summary statistics for the minimum wage dataset.}
\label{tab:minwage_emptot_summary}
\begin{tabular}{llcc}
\toprule
Outcome & & Pre ($\mathrm{TIME}=0$) & Post ($\mathrm{TIME}=1$) \\
\midrule
\texttt{emptot} & Control & 23.331 & 21.166 \\
                & Treated & 20.439 & 21.027 \\[0.3em]
\texttt{wage\_st} & Control & 4.630 & 4.618 \\
                  & Treated & 4.612 & 5.081 \\[0.3em]
\texttt{pmeal}  & Control & 3.042 & 3.027 \\
                & Treated & 3.351 & 3.415 \\
\bottomrule
\end{tabular}
\end{table}

\begin{table}[H]
\centering
\caption{Cross-method DiD inference for the minimum wage dataset across three outcomes.}
\label{tab:minwage_compare}
\begin{tabular}{llcc}
\toprule
Outcome & Method & $p$-value & Decision ($\alpha=0.05$) \\
\midrule
\multirow{5}{*}{\texttt{emptot} ($\hat{\delta}=2.754$)}
 & Doubly randomised & $<0.05$ & Rejected \\
 & Single randomised & $<0.05$ & Rejected \\
 & HC0 $t$-test & $<0.05$ & Rejected \\
 & CRVE1 $t$-test & $<0.05$ & Rejected \\
 & Hansen jackknife & $<0.05$ & Rejected \\
\midrule
\multirow{5}{*}{\texttt{wage\_st} ($\hat{\delta}=0.481$)}
 & Doubly randomised & $<0.001$ & Rejected \\
 & Single randomised & $<0.001$ & Rejected \\
 & HC0 $t$-test & $<0.001$ & Rejected \\
 & CRVE1 $t$-test & $<0.001$ & Rejected \\
 & Hansen jackknife & $<0.001$ & Rejected \\
\midrule
\multirow{5}{*}{\texttt{pmeal} ($\hat{\delta}=0.079$)}
 & Doubly randomised & $>0.05$ & Not rejected \\
 & Single randomised & $>0.05$ & Not rejected \\
 & HC0 $t$-test & $>0.05$ & Not rejected \\
 & CRVE1 $t$-test & $>0.05$ & Not rejected \\
 & Hansen jackknife & $>0.05$ & Not rejected \\
\bottomrule
\end{tabular}
\end{table}

All five methods reach identical conclusions for all three outcomes. The doubly randomised test is consistent with the standard approaches, including the jackknife, confirming that the proposed method does not lead to different empirical decisions in this well-powered dataset.

\begin{figure}[H]
\centering
\begin{tabular}{cc}
    \includegraphics[width=0.45\textwidth]{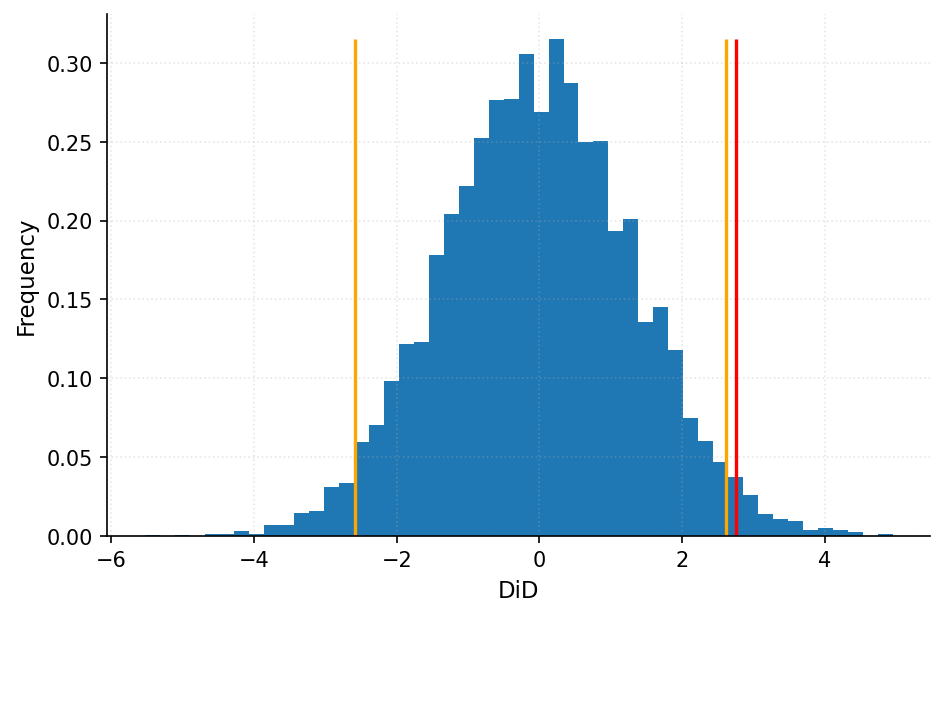} &
    \includegraphics[width=0.45\textwidth]{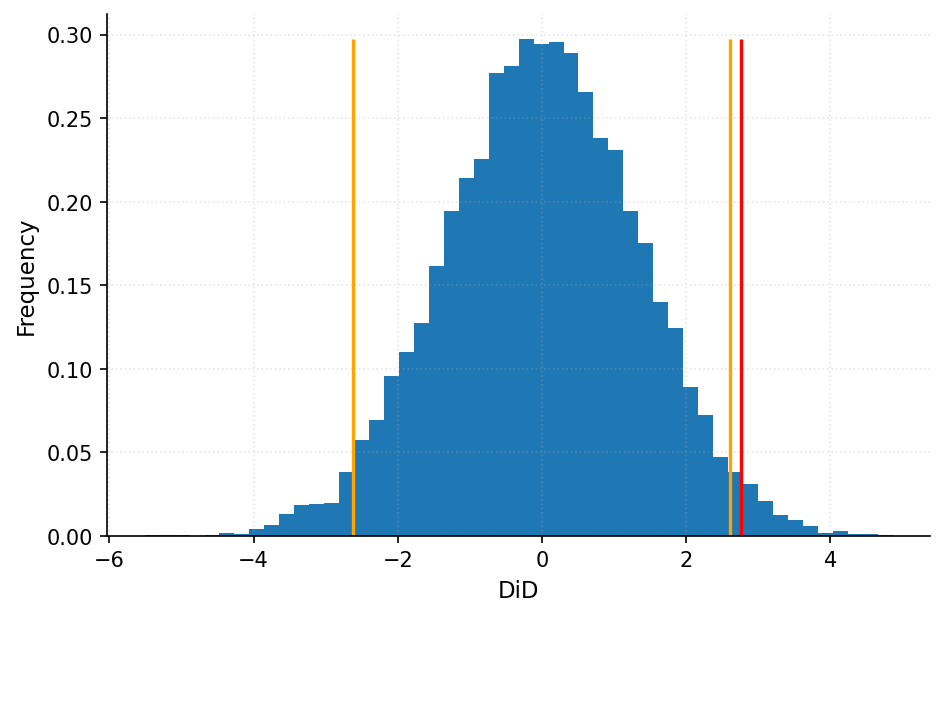} \\
    \textit{emptot, single rand.} & \textit{emptot, doubly rand.} \\[1em]
    \includegraphics[width=0.45\textwidth]{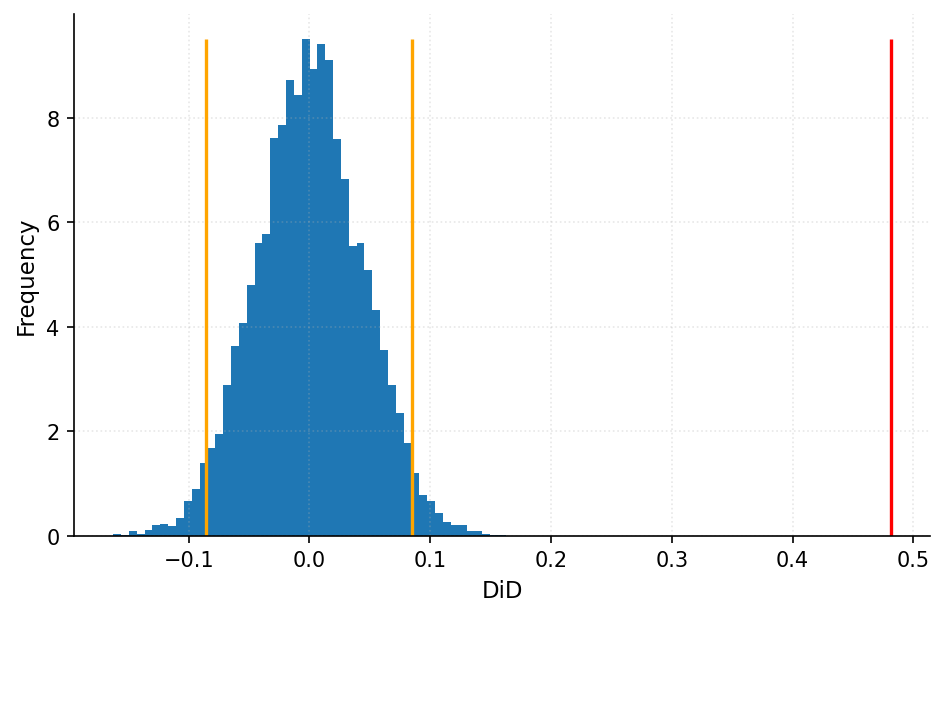} &
    \includegraphics[width=0.45\textwidth]{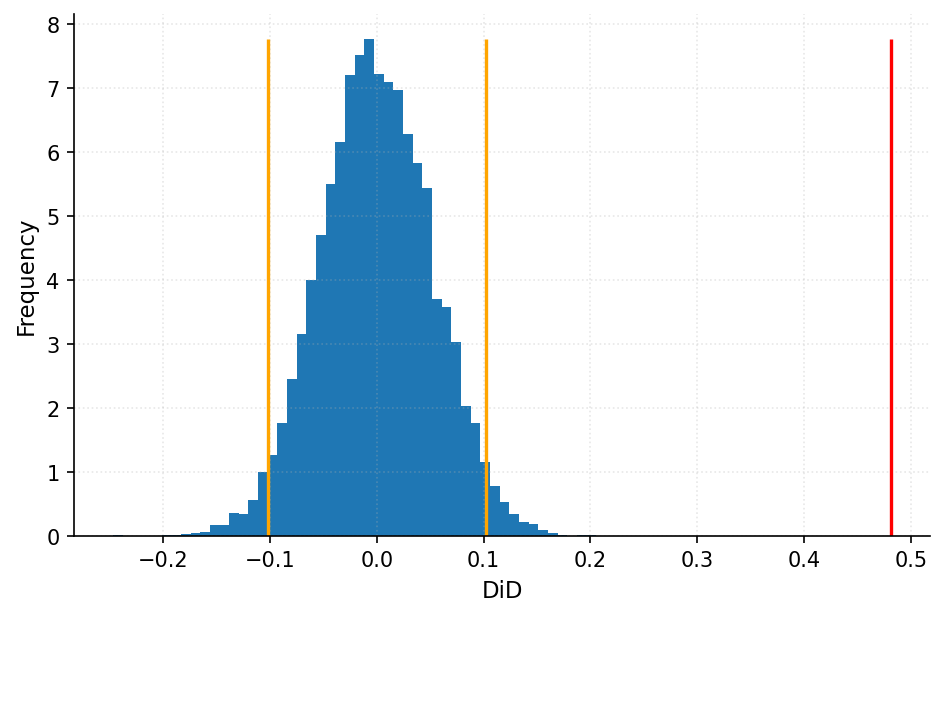} \\
    \textit{wage\_st, single rand.} & \textit{wage\_st, doubly rand.} \\[1em]
    \includegraphics[width=0.45\textwidth]{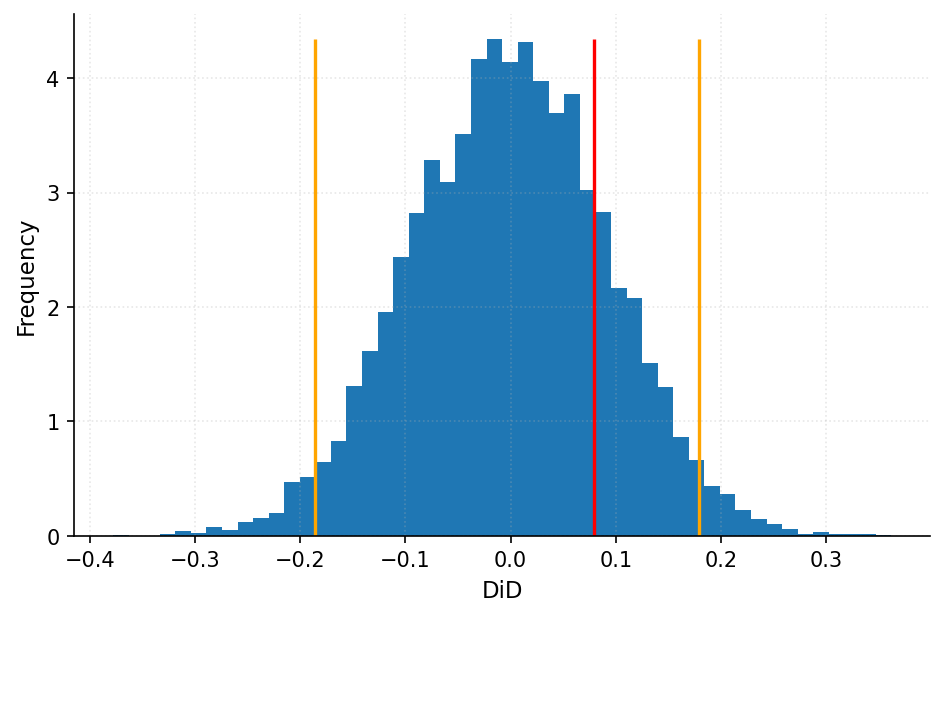} &
    \includegraphics[width=0.45\textwidth]{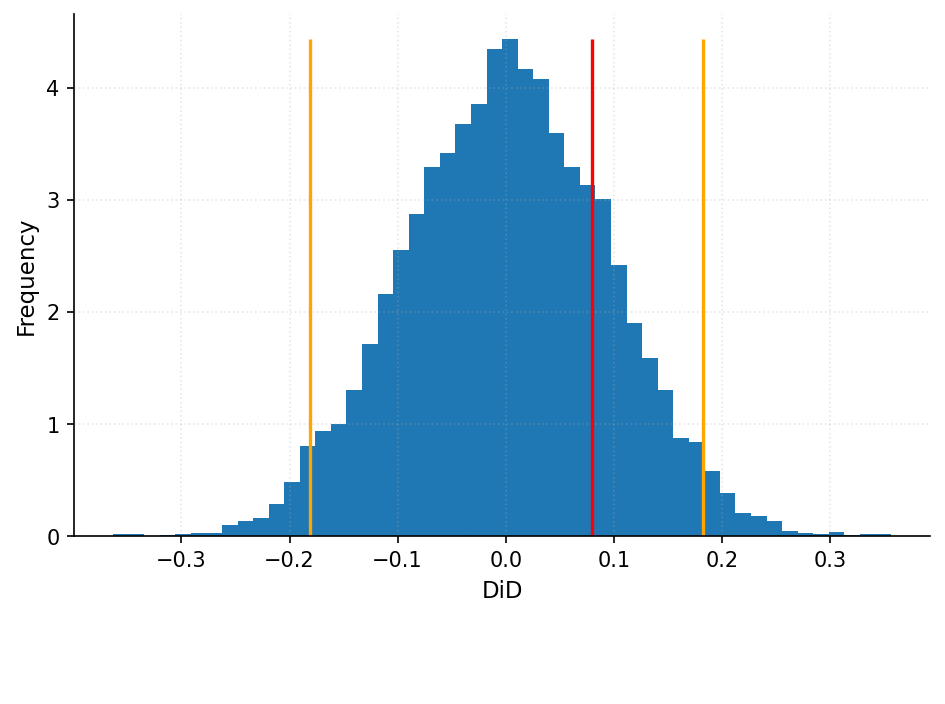} \\
    \textit{pmeal, single rand.} & \textit{pmeal, doubly rand.}
\end{tabular}
\caption{Null permutation distributions for the minimum wage dataset under single (left) and doubly randomised (right) inference.}
\label{fig:minwage_histograms}
\end{figure}

\subsubsection{Dataset 4: Refugee Arrivals and Far-Right Voting, Greece (2012–2016)}
\label{subsec:dinas_data}

The final dataset \citep{Dinas2019} covers 96 Greek municipalities over four parliamentary elections, with treatment corresponding to substantial refugee-inflow exposure. Table~\ref{tab:dinas_sample_summary} presents group means and Table~\ref{tab:dinas_compare} the cross-method comparison.

\begin{table}[H]
\centering
\caption{Summary statistics for the refugee-arrival dataset.}
\label{tab:dinas_sample_summary}
\begin{tabular}{lcc}
\toprule
 & Pre ($\mathrm{TIME}=0$) & Post ($\mathrm{TIME}=1$) \\
\midrule
Control ($\mathrm{AFFECTED}=0$) & 5.072 & 5.659 \\
Treated ($\mathrm{AFFECTED}=1$) & 5.730 & 8.404 \\
\bottomrule
\end{tabular}
\end{table}

\begin{table}[H]
\centering
\caption{Cross-method DiD inference for the refugee-arrival dataset ($\hat{\delta}=2.087$).}
\label{tab:dinas_compare}
\begin{tabular}{lcc}
\toprule
Method & $p$-value & Decision ($\alpha=0.05$) \\
\midrule
Doubly randomised & $<0.001$ & Rejected \\
Single randomised & $<0.001$ & Rejected \\
HC0 $t$-test & $<0.001$ & Rejected \\
CRVE1 $t$-test & $<0.001$ & Rejected \\
Hansen jackknife & $<0.001$ & Rejected \\
\bottomrule
\end{tabular}
\end{table}

All five methods strongly and consistently reject the null. The doubly randomised inference is broadly in agreement with the parametric alternatives and the single-margin test.

\begin{figure}[H]
    \centering
    \begin{tabular}{cc}
        \includegraphics[width=0.45\textwidth]{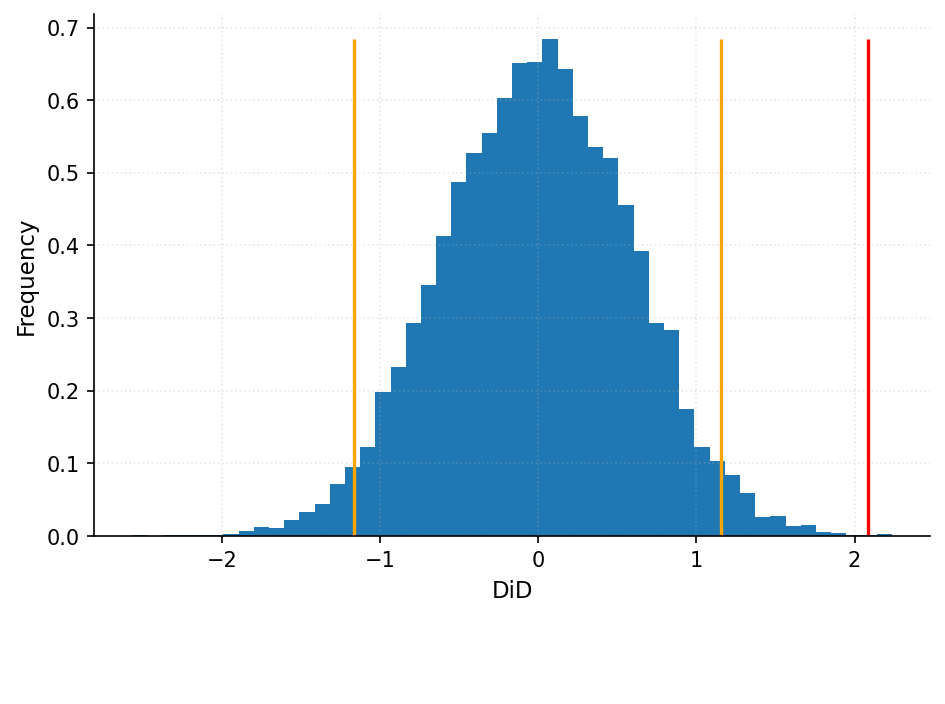} &
        \includegraphics[width=0.45\textwidth]{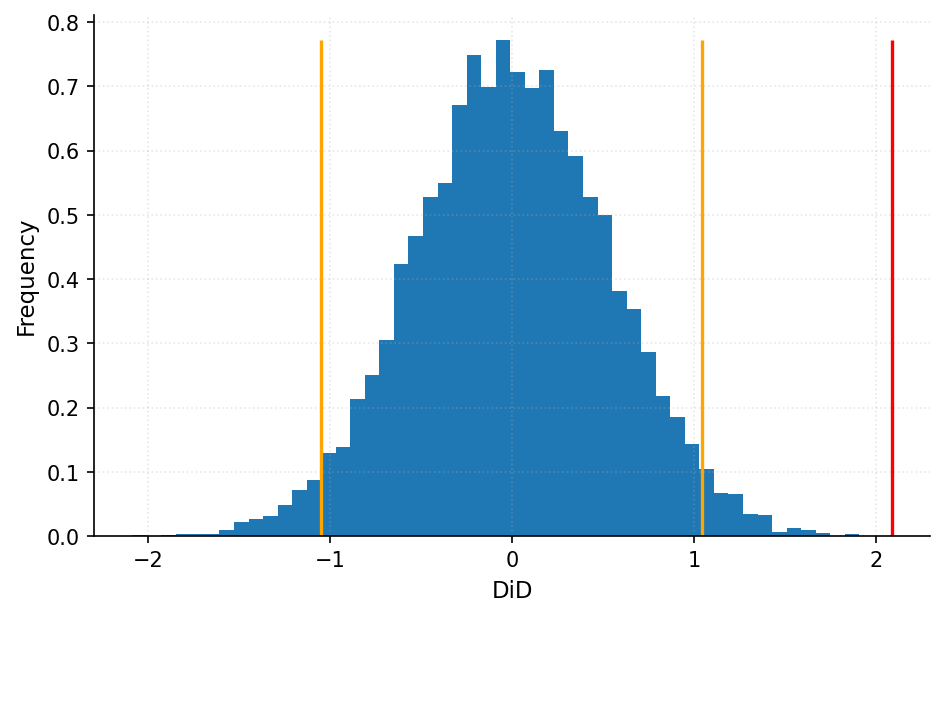}
    \end{tabular}
    \caption{Null permutation distributions for the refugee-arrival dataset under single (left) and doubly randomised (right) inference.}
    \label{fig:dinas_histograms}
\end{figure}

\subsection{Summary of empirical applications}
Across all four datasets the doubly randomised test reaches the same inference decision as HC0, CRVE1, and the Hansen jackknife. This is expected: in these datasets the sample sizes range from moderate to large and the effects that are significant are strongly so. The value of the proposed method is not to overturn well-established results in well-powered settings, but to provide reliable inference in the small-$n$, few-group configurations where the simulation study demonstrates that standard parametric methods can be substantially misleading.

\section{Conclusions}
\label{sec:conclusions}

This paper has developed and evaluated a doubly randomised significance test for the DiD estimator, positioning it explicitly within the econometric inference literature. Three broad conclusions emerge.

\paragraph{Standard methods can be severely unreliable at small $n$.}
Our simulation evidence confirms the warnings in \citet{Bertrand2004}, \citet{Donald2007}, and \citet{Hansen2025}: HC0 and CRVE1 $t$-tests reject the null at rates of $11.4\%$ and $8.2\%$ respectively at $n=20$ under the null, with inflation persisting through $n=60$. Any apparent power advantage of these tests at small $n$ is partly illusory --- a test that over-rejects under the null will always appear more powerful than one that does not. Applied researchers working with small panels or few treatment groups should treat HC- and CRVE-based inference with caution, and should prefer either the Hansen jackknife or a randomization-based test.

\paragraph{Parametric inflation is design-driven, not distribution-driven.}
The robustness experiments of Appendix~\ref{app:het} establish that the small-sample size distortion of HC0 and CRVE1 is essentially unchanged when the errors are heavy-tailed, skewed, or heteroskedastic. The source of the problem is the leverage and effective-degrees-of-freedom structure of the regressor matrix, in line with the diagnosis of \citet{Hansen2025}. Heteroskedasticity-robust standard errors thus correct for a structural feature of the errors that is, in DiD designs with few treated observations, of second-order importance compared with the leverage-driven bias in the variance estimator. The randomization-based test, whose validity rests on the maintained assignment mechanism rather than on any property of $\varepsilon$, is robust to both classes of departure by Theorem~\ref{thm:exactness}; together with the Hansen jackknife, it is one of two procedures in our comparison that controls size accurately under every simulation design we considered.

\paragraph{The doubly randomised test provides honest inference.}
In our simulations, the proposed test remains close to nominal size even at $n=20$ while operating on a permutation principle that requires neither distributional assumptions nor a cluster specification. The power cost relative to size-equivalent parametric tests is real but bounded: it is most visible at very small $n$ and large effect sizes, and effectively disappears by $n=100$ for practically relevant effects ($\delta \ge 0.75\sigma$). Among size-valid tests, single randomisation dominates slightly in raw power, since fixing the \texttt{TIME} margin concentrates the null distribution more tightly; the doubly randomised test trades this small power loss for the combinatorial richness of Proposition~\ref{prop:perm-gain}.

\paragraph{Entropy justification for balanced randomization.}
The choice of Bernoulli($1/2$) reshuffling for both margins is not arbitrary. By Proposition~\ref{prop:perm-gain}, dual randomization expands the permutation space by $\binom{n}{n_T}$. By Stirling's approximation \citep{Morris2024}, the logarithmic growth rate of this expansion is approximately $nH(p_T)$ where $H(p) = -p\log p - (1-p)\log(1-p)$ is the binary Shannon entropy \citep{Saraiva2023}. This is maximised at $p_T = 1/2$, giving
\[
\#\Omega(\text{AFFECTED \& TIME}) \asymp \frac{2^{2n}}{2\pi n}.
\]
The combined permutation universe grows exponentially in $n$, and balanced randomization maximises the information content of the null distribution, formally explaining the observed stability of the test across datasets and sample sizes.

\paragraph{Practical recommendations.}
We recommend the following workflow for practitioners:
\begin{enumerate}
\item When $n \le 100$ or when the number of groups or treated clusters is small, we recommend reporting the doubly randomised $p$-value alongside conventional parametric alternatives, and treating it as the primary robustness check for finite-sample inference.
\item For comparison and transparency, also report the Hansen jackknife $p$-value, which is the best parametric alternative in these settings \citep{Hansen2025}.
\item In larger samples where computational cost dominates, CRVE1 or the jackknife are reliable; there is no need for permutation tests when $n \ge 200$ and $G$ is large.
\item The \texttt{sigDD} R package (GitHub: \href{https://github.com/profsms/sigDD}{profsms/sigDD}) implements the doubly randomised test with a single function call.
\end{enumerate}

\paragraph{Directions for future research.}
Several extensions merit further investigation. First, characterising the power of the doubly randomised test analytically — rather than by simulation — would provide tighter guidance on the sample sizes where the method is preferable to parametric alternatives. Second, extensions to staggered treatment timing and heterogeneous treatment effects, as in \citet{Callaway2022} and \citet{Roth2023}, are natural but require adapting the permutation space. Third, the current procedure tests the sharp null; extending to approximate sharp nulls or average treatment effect nulls, following \citet{Canay2017}, would broaden applicability. Finally, standardizing the permutation distribution to remove sample-size dependence remains an open theoretical question.

\section*{Acknowledgments}
We are grateful to the anonymous reviewer at the \emph{Central European Journal of Economic Modelling and Econometrics}, whose suggestion to examine the test's behaviour under more realistic error structures led us to the robustness experiments of Appendix~\ref{app:het} and clarified our diagnosis that the small-sample distortion of CRVE-type procedures is design-driven rather than distribution-driven.\\
One of the authors also thanks Matthias Ginger for insightful discussions on permutation-based inference and for pointing out related approaches in other scientific disciplines.

\newpage
\bibliographystyle{erae}
\bibliography{references}

\appendix
\section{Pseudocode for the doubly randomised significance test}
\label{app:pseudocode}

\begin{algorithmic}[2]
	\State $Y$: observed outcome vector
	\State $\mathrm{Time}$: binary vector ($0$ pre-treatment, $1$ post-treatment)
	\State $\mathrm{Affected}$: binary vector ($0$ control, $1$ treated)
	\State $\mathrm{DDvalues}$: empty vector
	\State $n_Y$: length of $Y$
	\State $B$: number of permutations (recommend $\ge 4{,}999$)
	\State $\alpha$: nominal significance level
	\Function{CalculateDD}{$M, V_1, V_2$}
	\State $DD_c \gets$ interaction coefficient from OLS of $M$ on $(1, V_1, V_2, V_1 \cdot V_2)$
	\State \Return $DD_c$
	\EndFunction
	
	\Function{RandBern}{$n, p$}
	\State \Return vector of $n$ independent Bernoulli$(p)$ draws
	\EndFunction
	
	\For{$b \gets 1$ to $B$}
	\State $X_1 \gets$ \Call{RandBern}{$n_Y, 1/2$}
	\State $X_2 \gets$ \Call{RandBern}{$n_Y, 1/2$}
	\State $DD_b \gets$ \Call{CalculateDD}{$Y, X_1, X_2$}
	\State Append $DD_b$ to $\mathrm{DDvalues}$
	\EndFor
	
	\State $DD_e \gets$ \Call{CalculateDD}{$Y, \mathrm{Time}, \mathrm{Affected}$}
	\State $\hat{p} \gets (1 + \#\{b : |DD_b| \ge |DD_e|\}) / (B + 1)$
	\If{$\hat{p} \le \alpha$}
	\State Reject $H_0$: DiD estimator is statistically significant
	\Else
	\State Fail to reject $H_0$
	\EndIf
\end{algorithmic}

\section{Empirical Power on Subsamples of the Ben \& Jerry's Dataset}
\label{app:subsampling}

The simulation study of Section~\ref{sec:numerics} establishes the finite-sample properties of the doubly randomised test under a controlled Gaussian DGP. This appendix asks a complementary question: how does the test's power behave when applied to small random subsamples of a real empirical dataset where the full-sample effect is known to be strongly significant?

\subsection{Design}

We use the Ben~\&~Jerry's vs.\ Häagen-Dazs dataset from \citet{Halkiewicz2024thesis} ($n_{\mathrm{full}} = 370$, $\hat{\delta} = 4.827$, full-sample $\hat{p} < 0.001$). For each subsample size $n \in \{20, 30, 50, 100, 150, 200\}$, we draw $R = 1{,}000$ independent random subsamples without replacement and apply the doubly randomised test (OLS-based, $B = 499$ permutations, $\alpha = 0.05$) to each. The rejection rate across the $R$ draws is an empirical estimate of the test's power at that sample size on this dataset. The experiment uses seed 2026 throughout and is exactly reproducible from the \texttt{sigDD} package data.

\subsection{Results}

\begin{table}[H]
\centering
\caption{Empirical rejection rates of the doubly randomised test on $R=1{,}000$ subsamples of the Ben~\&~Jerry's dataset. Wilson 95\% confidence intervals in brackets.}
\label{tab:bj_subsampling}
\begin{tabular}{ccc}
\toprule
$n$ & Rejection rate & 95\% CI \\
\midrule
20  & 0.027 & [0.019,\ 0.039] \\
30  & 0.049 & [0.037,\ 0.064] \\
50  & 0.135 & [0.115,\ 0.158] \\
100 & 0.410 & [0.380,\ 0.441] \\
150 & 0.716 & [0.687,\ 0.743] \\
200 & 0.941 & [0.925,\ 0.954] \\
370 & 1.000 & [full sample, $\hat{p} = 0.001$] \\
\bottomrule
\end{tabular}
\end{table}

\begin{figure}[H]
\centering
\includegraphics[width=0.82\textwidth]{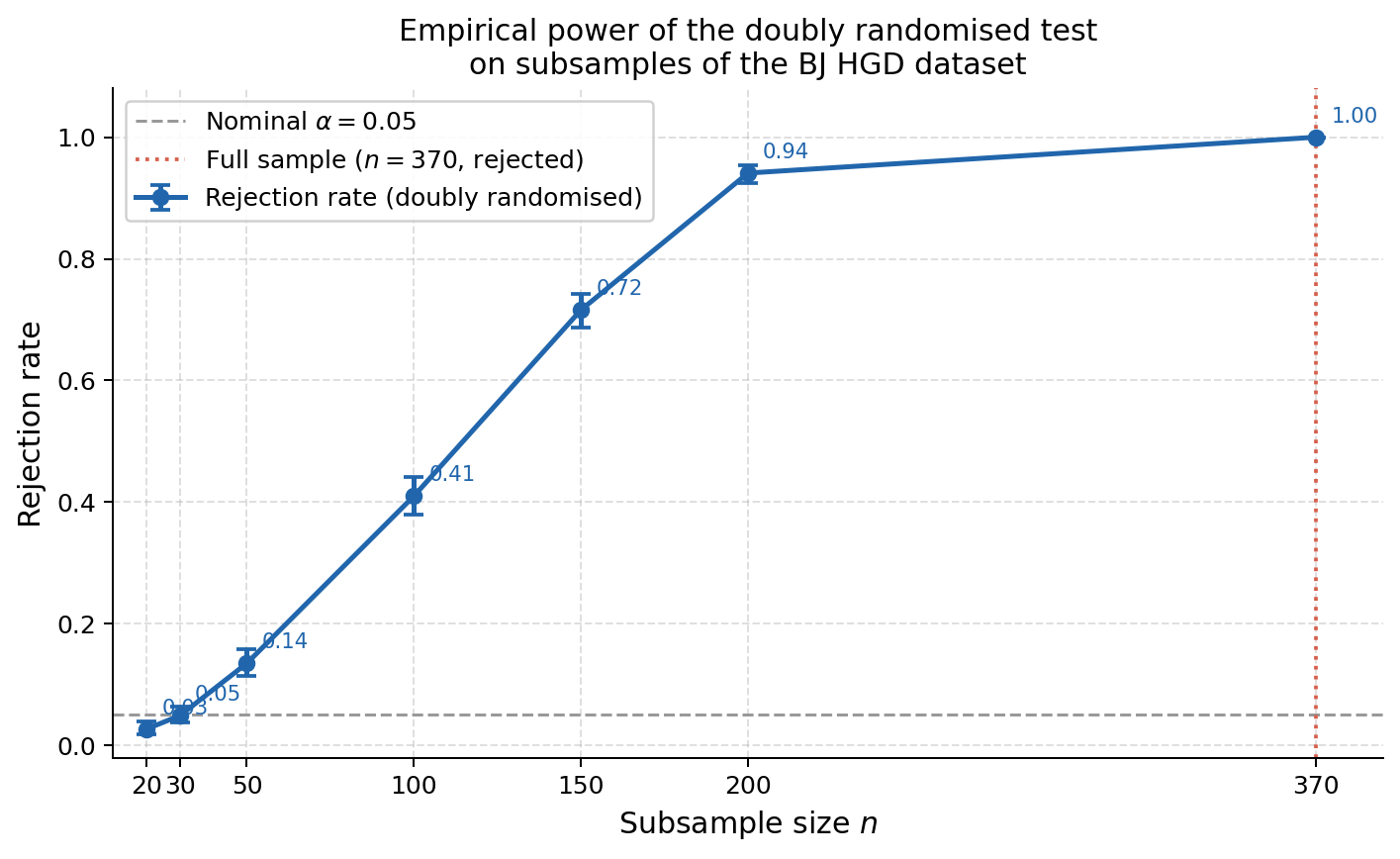}
\caption{Empirical rejection rate of the doubly randomised test as a function of subsample size $n$, with Wilson 95\% confidence intervals. The dashed line marks the nominal $\alpha = 0.05$; the dotted vertical line marks the full sample ($n=370$). Based on $R=1{,}000$ subsamples per $n$, $B=499$ permutations each.}
\label{fig:bj_power_curve}
\end{figure}

\subsection{Interpretation}

Three features of Figure~\ref{fig:bj_power_curve} and Table~\ref{tab:bj_subsampling} are worth noting.

First, at $n = 20$ and $n = 30$ the rejection rates are $2.7\%$ and $4.9\%$ — at or below the nominal $5\%$. Since the full-sample effect is strongly real, these low rates confirm that the test is not over-rejecting in small samples, in line with the size-control guarantee of Theorem~\ref{thm:exactness}. The test pays a power cost for its exact size control, and this cost is visible here.

Second, power grows monotonically and steeply with $n$. The rate reaches $41\%$ at $n = 100$, $71.6\%$ at $n = 150$, and $94.1\%$ at $n = 200$ — more than half the full dataset. The standardised effect size is moderate ($\hat{\delta}/\hat{\sigma} \approx 0.64$), so these power levels are consistent with what Lemma~\ref{lem:density} predicts: as the permutation space grows, the null distribution becomes a better approximation of the theoretical $F_{H_0}$ and the test gains resolution.

Third, the shape of the power curve mirrors that of the doubly randomised test in the synthetic simulation at $\delta = 0.5\sigma$ (Figure~\ref{fig:power_grid}), providing external consistency between the controlled and empirical evidence.

\section{Empirical Size — Numerical Table}
\label{app:size_table}

Table~\ref{tab:size} provides the numerical values underlying Figure~\ref{fig:size} in the main text, together with additional sample sizes that densify the grid at small $n$.

\begin{table}[H]
\centering
\caption{Empirical size at $\alpha = 0.05$ across procedures and sample sizes
         ($N_{\mathrm{outer}} = 5{,}000$ replications, $B = 4{,}999$, seed 2026).
         Values above 0.065 indicate liberal tests and are shown in red.}
\label{tab:size}
\begin{tabular}{lccccc}
\toprule
$n$ & Doubly rand. & Single rand. & HC0 & CRVE1 & Hansen jack. \\
\midrule
20  & 0.043 & 0.046 & \textcolor{red}{0.114} & \textcolor{red}{0.082} & 0.031 \\
30  & 0.039 & 0.046 & \textcolor{red}{0.076} & 0.060 & 0.038 \\
40  & 0.045 & 0.051 & \textcolor{red}{0.074} & 0.058 & 0.042 \\
50  & 0.038 & 0.042 & 0.054 & 0.046 & 0.037 \\
60  & 0.044 & 0.050 & 0.061 & 0.055 & 0.047 \\
80  & 0.041 & 0.047 & 0.056 & 0.052 & 0.044 \\
100 & 0.036 & 0.043 & 0.052 & 0.046 & 0.041 \\
150 & 0.042 & 0.049 & 0.055 & 0.052 & 0.048 \\
200 & 0.042 & 0.048 & 0.053 & 0.050 & 0.048 \\
400 & 0.038 & 0.049 & 0.052 & 0.051 & 0.050 \\
\bottomrule
\end{tabular}
\end{table}

The densified small-$n$ grid confirms and sharpens the picture established by the original six sample sizes. The randomization tests and the Hansen jackknife stay below the liberal threshold of $0.065$ at every $n$, while HC0 remains anti-conservative through $n = 40$ and CRVE1 has settled close to the nominal level by $n = 30$. The transition of HC0 from liberal to acceptable behaviour is gradual rather than abrupt: rejection rates decay smoothly from $0.114$ at $n = 20$ through $0.076$ at $n = 30$ and $0.074$ at $n = 40$, crossing below the liberal threshold by $n = 50$ (rate $0.054$), after which values fluctuate within Monte Carlo noise of the nominal level.

\section{Robustness to Departures from i.i.d.\ Gaussian Errors}
\label{app:het}

The simulation of Section~\ref{subsec:simdesign} uses i.i.d.\ $N(0,1)$ errors --- a strong joint assumption combining normality and homoskedasticity. This appendix examines how the five procedures behave under two classes of departure: (D.1) homoskedastic but non-Gaussian errors, and (D.2) heteroskedastic errors. All other design parameters are as in Section~\ref{subsec:simdesign}; we report size at $\delta = 0$ only.

\subsection{Non-Gaussian homoskedastic errors}
\label{app:het:nongauss}

We retain $\mathrm{Var}(\varepsilon_i) = 1$ and i.i.d.\ sampling, but replace the Gaussian with two alternatives:
\begin{itemize}
  \item \textbf{NG1} (heavy tails): $\varepsilon_i \sim t_3 / \sqrt{3}$, so $\mathrm{Var}(\varepsilon_i) = 1$ with substantially heavier tails than $N(0,1)$.
  \item \textbf{NG2} (skewness): $\varepsilon_i = (\chi^2_3 - 3) / \sqrt{6}$, a centred and rescaled $\chi^2_3$ with $\mathrm{Var}(\varepsilon_i) = 1$ and pronounced right-skew.
\end{itemize}
The randomisation tests are validity-preserving in this setting by construction (exchangeability of labels under the sharp null does not require Gaussianity), so the experiment serves primarily as a diagnostic for the parametric procedures.

\begin{table}[H]
\centering
\caption{Empirical size under non-Gaussian homoskedastic errors ($\alpha = 0.05$, $N_{\mathrm{outer}} = 5{,}000$, $B = 4{,}999$). Values above 0.065 in red.}
\label{tab:nongauss}
\begin{tabular}{llccccc}
\toprule
Design & $n$ & Doubly rand. & Single rand. & HC0 & CRVE1 & Hansen jack. \\
\midrule
\multirow{4}{*}{NG1 ($t_3$)}  & 20  & 0.040 & 0.048 & \textcolor{red}{0.105} & \textcolor{red}{0.071} & 0.027 \\
                              & 30  & 0.035 & 0.041 & 0.062 & 0.046 & 0.030 \\
                              & 50  & 0.038 & 0.048 & 0.061 & 0.051 & 0.040 \\
                              & 100 & 0.037 & 0.044 & 0.051 & 0.046 & 0.041 \\
\midrule
\multirow{4}{*}{NG2 ($\chi^2_3$)} & 20  & 0.044 & 0.052 & \textcolor{red}{0.113} & \textcolor{red}{0.081} & 0.033 \\
                                  & 30  & 0.036 & 0.048 & \textcolor{red}{0.073} & 0.056 & 0.037 \\
                                  & 50  & 0.039 & 0.045 & 0.057 & 0.048 & 0.035 \\
                                  & 100 & 0.039 & 0.042 & 0.049 & 0.045 & 0.040 \\
\bottomrule
\end{tabular}
\end{table}

The randomization tests stay at or below nominal across all eight cells, with maximum rejection rate $0.052$. This is the behaviour predicted by Theorem~\ref{thm:exactness}: validity does not depend on the error distribution. The Hansen jackknife is similarly insulated, remaining below $0.05$ throughout. HC0 is anti-conservative at $n=20$ under both heavy tails (NG1, $0.105$) and skewness (NG2, $0.113$); under NG2 the inflation persists to $n=30$ ($0.073$), while under NG1 it has already abated by that sample size. CRVE1 is mildly anti-conservative at $n = 20$ ($0.071$ under NG1, $0.081$ under NG2). Heavy tails ($t_3$) produce slightly \emph{less} HC0 inflation at $n = 20$ than the Gaussian baseline of Table~\ref{tab:size} ($0.105$ vs $0.114$), suggesting that occasional large residuals partly compensate for the small-sample downward bias in the variance estimator. By $n = 100$ all parametric procedures recover near-nominal size.

\subsection{Heteroskedastic errors}
\label{app:het:het}

Three heteroskedastic departures preserve normality but break the constant-variance assumption. In each case, $\varepsilon_i \sim N(0, \sigma_i^2)$ with $\sigma_i^2$ depending on the binary treatment indicators:
\begin{itemize}
  \item \textbf{H1} (variance scales with treatment status): $\sigma_i^2 = 1 + \kappa\,\mathrm{AFFECTED}_i$.
  \item \textbf{H2} (variance scales with time period): $\sigma_i^2 = 1 + \kappa\,\mathrm{TIME}_i$.
  \item \textbf{H3} (variance scales with treatment exposure): $\sigma_i^2 = 1 + \kappa\,\mathrm{TIME}_i \times \mathrm{AFFECTED}_i$.
\end{itemize}
We set $\kappa = 2$, so that the variance triples in the affected cell. This is a substantial but empirically plausible degree of heteroskedasticity.

\begin{table}[H]
\centering
\caption{Empirical size under three heteroskedastic error designs ($\kappa = 2$, $\alpha = 0.05$, $N_{\mathrm{outer}} = 5{,}000$, $B = 4{,}999$). Values above 0.065 in red.}
\label{tab:het}
\begin{tabular}{llccccc}
\toprule
Design & $n$ & Doubly rand. & Single rand. & HC0 & CRVE1 & Hansen jack. \\
\midrule
\multirow{4}{*}{H1} & 20  & 0.053 & 0.060 & \textcolor{red}{0.119} & \textcolor{red}{0.089} & 0.038 \\
                    & 30  & 0.047 & 0.056 & \textcolor{red}{0.084} & \textcolor{red}{0.066} & 0.042 \\
                    & 50  & 0.049 & 0.051 & 0.063 & 0.052 & 0.042 \\
                    & 100 & 0.048 & 0.053 & 0.059 & 0.054 & 0.049 \\
\midrule
\multirow{4}{*}{H2} & 20  & 0.050 & 0.051 & \textcolor{red}{0.121} & \textcolor{red}{0.091} & 0.039 \\
                    & 30  & 0.044 & 0.047 & \textcolor{red}{0.080} & 0.061 & 0.037 \\
                    & 50  & 0.050 & 0.049 & \textcolor{red}{0.069} & 0.057 & 0.046 \\
                    & 100 & 0.043 & 0.048 & 0.056 & 0.049 & 0.043 \\
\midrule
\multirow{4}{*}{H3} & 20  & 0.047 & 0.051 & \textcolor{red}{0.119} & \textcolor{red}{0.084} & 0.033 \\
                    & 30  & 0.049 & 0.053 & \textcolor{red}{0.084} & \textcolor{red}{0.065} & 0.040 \\
                    & 50  & 0.050 & 0.055 & \textcolor{red}{0.072} & 0.059 & 0.046 \\
                    & 100 & 0.045 & 0.049 & 0.055 & 0.051 & 0.047 \\
\bottomrule
\end{tabular}
\end{table}

Across all twelve cells in Table~\ref{tab:het}, the randomization tests stay below the liberal threshold of $0.065$, with maximum rejection rate $0.060$ (single rand, H1 at $n=20$). The Hansen jackknife is conservative throughout ($\le 0.049$). HC0 and CRVE1 inflate at small $n$ (HC0 reaches $0.121$ under H2 at $n=20$) with magnitudes at $n \le 30$ comparable to those under the homoskedastic Gaussian baseline of Table~\ref{tab:size}. At $n = 50$, however, HC0 inflation persists under the two designs (H2 and H3) in which the variance heterogeneity interacts with the time margin, while it has already abated under H1. This pattern is consistent with the leverage-driven diagnosis of Section~\ref{sec:inference_lit}: heteroskedasticity \emph{along the same margin} as a high-leverage regressor amplifies the small-sample bias in the variance estimator, but pure variance-by-treatment-status heterogeneity (H1) does not. By $n = 100$ all parametric procedures recover near-nominal size. The randomization tests, whose validity is by construction independent of error variance heterogeneity, are unaffected throughout.

\subsection{Summary of robustness experiments}

Aggregating across all five departure designs and four sample sizes (40 cells per method): the doubly randomised and single randomised tests stay at or below the liberal threshold of $0.065$ in every cell, and the Hansen jackknife remains conservative in every cell. HC0 is liberal in $11$ of $20$ cells and CRVE1 in $7$ of $20$. The inflation is concentrated at $n \le 30$, with two heteroskedastic cells (H2 and H3 at $n = 50$) showing residual HC0 over-rejection consistent with the leverage-driven diagnosis: variance heterogeneity along the time margin amplifies the small-sample bias that is already present under the iid Gaussian baseline. Crucially, the parametric procedures' difficulties are visible already under the strongest distributional assumptions (Table~\ref{tab:size}), and the magnitudes of inflation at $n \le 30$ are comparable across the Gaussian baseline and the five departure designs. This is consistent with the diagnosis in Section~\ref{sec:inference_lit}: the small-sample distortion of CRVE-type procedures is driven by leverage and effective degrees of freedom of the regressor matrix, not by the error distribution. The randomization-based tests, validated by Theorem~\ref{thm:exactness}, are insulated from both classes of departure by construction.
\end{document}